\documentclass[a4paper]{article}

\usepackage{amsmath}
\usepackage{amssymb}
\usepackage{mathtools}
\usepackage{amsthm}
\usepackage{physics}
\usepackage{subfig}
\usepackage{cite}
\usepackage{braket}
\usepackage{xcolor}
\usepackage{booktabs}
\usepackage{tikz,tikz-3dplot}
\usepackage{pgfplots}
\usetikzlibrary{arrows}
\usetikzlibrary{patterns}

\newcommand{\half}{\frac{1}{2}}

\newcommand{\N}{\mathbb{N}}
\newcommand{\R}{\mathbb{R}}

\newcommand{\C}{\mathbb{C}}
\newcommand{\Q}{\mathbb{Q}}
\newcommand{\ud}{\mathrm{d}}

\DeclareMathOperator{\diag}{diag}

\newcommand{\BigO}[1]{\mathcal{O}\left( #1 \right)}
\newcommand{\xv}{\mathbf{x}}
\newcommand{\xvd}{\mathbf{\dot{x}}}
\newcommand{\xvdd}{\mathbf{\ddot{x}}}
\renewcommand{\vec}{\mathbf}

\theoremstyle{plain}
\newtheorem{theorem}{Theorem}

\newtheorem{corollary}[theorem]{Corollary}
\newtheorem{lemma}[theorem]{Lemma}
\theoremstyle{definition}

\theoremstyle{remark}
\newtheorem{remark}{Remark}

\renewcommand{\imath}{\mathrm{i}}
\renewcommand{\epsilon}{\varepsilon}
\renewcommand{\phi}{\varphi}

\allowdisplaybreaks
\makeatletter 
\renewcommand{\pdv}[2]{\begingroup 
\@tempswafalse\toks@={}\count@=\z@ 
\@for\next:=#2\do 
{\expandafter\check@var\next\@nil
 \advance\count@\der@exp 
 \if@tempswa 
   \toks@=\expandafter{\the\toks@\,}%
 \else 
   \@tempswatrue 
 \fi 
 \toks@=\expandafter{\the\expandafter\toks@\expandafter\partial\der@var}}%
\frac{\partial\ifnum\count@=\@ne\else^{\number\count@}\fi#1}{\the\toks@}%
\endgroup} 
\def\check@var{\@ifstar{\mult@var}{\one@var}} 
\def\mult@var#1#2\@nil{\def\der@var{#2^{#1}}\def\der@exp{#1}} 
\def\one@var#1\@nil{\def\der@var{#1}\chardef\der@exp\@ne} 
\makeatother
\title{A multiple scales approach to maximal superintegrability}

\author{G. Gubbiotti\textsuperscript{1,}\footnote{e-mail: \texttt{giorgio.gubbiotti@sydney.edu.au}} 
\and D. Latini\textsuperscript{2,}\footnote{e-mail: \texttt{latini@fis.uniroma3.it}}}

\date{\textsuperscript{1}School of Mathematics and Statistics, 
The University of Sydney, New South Wales 2006, Australia
\\
    \textsuperscript{2}Department of Mathematics and Physics, 
Roma Tre University, Via della Vasca Navale 84, I-00146 Rome, Italy}

\begin{document}

\maketitle

\begin{abstract}
    In this paper 
    we present a simple, algorithmic
    test to establish if a Hamiltonian system is maximally 
    superintegrable or not.
    This test is based on a very simple corollary of
    a theorem due to Nekhoroshev and on a perturbative
    technique called multiple scales method.
    If the outcome is positive, this test can be used
    to suggest maximal superintegrability, whereas when
    the outcome is negative it can be used to disprove it.
    This method can be regarded as a finite dimensional analog
    of the multiple scales method as a way to produce 
    soliton equations.
    We use this technique to show that the real
    counterpart of a mechanical system found by Jules Drach 
    in 1935 is, in general, not maximally superintegrable.
    We give some hints on how this approach could be applied
    to classify maximally superintegrable systems  by presenting a 
    direct proof of the well-known Bertrand's theorem.
\end{abstract}

\section{Introduction}

The notion of integrability in Classical Mechanics is well established 
since the times of Liouville \cite{Liouville1855} and, roughly speaking, means 
the existence of a ``sufficiently'' high number of \emph{integrals of motion}. 
To be more precise assume we are given a Hamiltonian
system with Hamiltonian $H=H\left( \vec{p},\vec{q} \right)$ with $n$ degrees
of freedom, i.e. $\vec{q}=\left( q_{1},\dots,q_{n} \right)$
and $\vec{p}=\left( p_{1},\dots,p_{n} \right)$.
We say that the mechanical system defined by the Hamiltonian
$H$ is integrable if there
exist $n$ \emph{integrals of motion}, i.e. $n$ functions $H_{i}$
$i=1,\dots,n$ which Poisson-commute with the Hamiltonian:
\begin{equation}
\left\{ H_{i},H \right\}=0.
\label{eq:poissoncomm}
\end{equation}
If an integral of motion is polynomial in $\vec{p}$, then its
total degree as polynomial in $\vec{p}$ is called the \emph{order}
of the integral of motion.
The Hamiltonian, which trivially commutes with itself, is included in
the list as $H_{1}=H$. 
These integrals of motion must be well defined functions
on the phase space, i.e. \emph{analytic} and \emph{single-valued}. 
Moreover, they have to be \emph{in involution}:
\begin{equation}
\left\{ H_{i},H_{j} \right\}=0, \quad i,j \in\{1,\dots,n\}.
\label{eq:involution}
\end{equation}
Finally, they must be functionally independent:
\begin{equation}
    \rank \frac{\partial (H_{1},\dots,H_{n})}{\partial(\vec{p},\vec{q})} = n.
\label{eq:rankcond}
\end{equation}
Indeed, the knowledge of these kinds of integrals permits the
integration of the equations of motion associated to $H$.
This is the content of the famous \emph{Liouville theorem} \cite{Liouville1855,Whittaker}.
We remark that the application of Liouville's theorem yields
other $n-1$ quantities that are Poisson-commuting with the Hamiltonian.
Usually, these quantities are not integrals of motion in the
sense of our definition since they are not well defined functions on the phase space. Integrability in Classical Mechanics implies that the motion is constrained on a subspace of the full phase space.
With some additional  assumptions on the geometric structure of the integrals of motion, it is possible to
prove that the motion is quasi-periodic on some tori in the phase space \cite{Arnold1967,Arnold1997}.

When there exist more than $n$, say $n+k$ with $1\leq k\leq n-1$,  independent integrals of motion, 
we say that the system is \emph{superintegrable}.
When $k=1$ the system is \emph{minimally superintegrable},
whereas when $k=n-1$ it is \emph{maximally superintegrable}. 
 In the case $n=2$ the two notions coincide.
The search for superintegrable systems started more than a half
century ago with the seminal paper \cite{Fris1965}, but the name
``superintegrability'' has been introduced only in \cite{Rauch1983}
about the Calogero-Moser system 
\cite{Calogero1971,Calogero1971erratum,Moser1975}.
However, we remark that the definition of superintegrability introduced
in \cite{Rauch1983} is in fact the definition of maximal superintegrability 
given above.
For a full historical and a state-of-art perspective on superintegrability
we refer to the review \cite{MillerPostWinternitz2013R} and references therein.

From the algebraic point of view the structure of superintegrable
systems is richer than that of integrable systems because
the additional integrals will not be in involution with the previous ones.
This gives rise to many interesting non-abelian algebraic structures.
Usually, they are finitely generated polynomial algebras, 
only exceptionally finite dimensional Lie algebras or Kac-Moody algebras
\cite{DabulSlodowyDabul1993}.
For this reason superintegrability is also often called non-abelian 
integrability.
From the geometrical point of view
superintegrability restricts trajectories to an $n - k$ dimensional subspace of the phase space.
This implies that the following Theorem holds true:
\begin{theorem}[Nekhoroshev \cite{Nekhoroshev1972}\footnote{This theorem
    is a particular case of Theorem 3 in \cite{Nekhoroshev1972}, which is sufficient for our discussion.}]
Let us consider an Hamiltonian system with Hamiltonian
$H = H\left( \vec{p},\vec{q} \right)$.
If such system is maximally superintegrable
then every bounded orbit is closed and periodic.
\label{th:orbits}
\end{theorem}
\noindent Intuitively this happens because in the maximally superintegrable
case the trajectories are restricted to one-dimensional subspace of the phase space. 
Therefore, any bounded orbit is just diffeomorphic to a circle.
Of course Theorem \ref{th:orbits} has an immediate corollary given by:
\begin{corollary}
    If a Hamiltonian system with Hamiltonian $H=H\left( \vec{p},\vec{q} \right)$
    possesses \emph{at least} one bounded orbit which is \emph{neither} 
    closed \emph{nor} periodic, then it is not maximally superintegrable.
    \label{cor:nonperiodic}
\end{corollary}

In this paper we suggest a simple, algorithmic, perturbative test
based on Corollary \ref{cor:nonperiodic} which allows to prove if
a system is maximally superintegrable or not.
This test can be particularly useful when $n=2$ since, as noted above, 
in this case the notion of
superintegrability and maximal superintegrability coincide.

The paper is structured as follows. 
In Section \ref{sec:method} we describe our method for disproving 
or suggesting maximal superintegrability.
We give a concise introduction to the \emph{multiple scales method} 
\cite{Kuzmak1959,ColeKevorkian1963}: 
a perturbative technique aimed to avoid resonant, i.e. diverging, 
terms in the asymptotic expansion and to describe physical
phenomena happening on different time scales. 
In Section \ref{sec:examples} we present some relevant examples
of application of the method.
In particular we discuss the field of applicability of the method 
and we confront it with other techniques.
We present the known example of the Tremblay-Turbiner-Winternitz (TTW) 
system \cite{TTW2009}, where the method of Section \ref{sec:method} 
suggests (maximal) superintegrability.
We also present a new result about the Drach system \cite{Drach1935}. 
Using the method presented in Section \ref{sec:method} we
show that the Drach system is, in general, not superintegrable.
Finally, we show how it is possible to use such a method
to classify maximally superintegrable systems giving an alternative proof 
of the famous Bertrand's Theorem \cite{Bertrand1873}.
In Section \ref{sec:conclusion} we give some conclusions and perspectives
for further developments.
We comment on the analogy between our method, 
which has been applied to finite dimensional systems, 
and the ones used in the infinite dimensional framework as a way to  produce soliton equations
\cite{ZakharovKuznetsov1986,CalogeroEchkhaus1987,CalogeroEchkhaus1988,Calogero1991}.

\section{The method}
\label{sec:method}

Our approach in disproving or suggesting maximal superintegrability
is based on Corollary \ref{cor:nonperiodic} and on the
so-called \emph{multiple scales method}.
The multiple scales analysis is a perturbation technique whose history dates back to the 18th century.
The bases for the multiple scales method were laid in the works by 
Lindstedt \cite{Lindstedt1882} and Poincar\'e \cite{Poincare1886}, 
it was developed in its modern form in \cite{Kuzmak1959,ColeKevorkian1963}.
The core of this approach is to find asymptotic approximated solutions of  a system of differential equations when the standard perturbation theory produces secular terms.
During the years, the multiple scales analysis has proved to be very useful in the construction of approximate solutions of differential equations, and is now included in every textbook on perturbation theory \cite{Nayfeh,BenderOrszag,KevorkianCole,Holmes}.
Such a powerful method has also found applications
in fields which do not seem to be correlated with such problems, for example
in the theory of integrable systems in infinite dimensions 
\cite{ZakharovKuznetsov1986,CalogeroEchkhaus1987,CalogeroEchkhaus1988,Calogero1991}. 

The key feature that allows the elimination of the secular terms
is the introduction of fast-scale variables and slow-scales variables
in a way that the dependence on the slow-scale variables will
prevent the secularities. To be more precise,
suppose that we are given a system of second-order ordinary differential
equations with independent variable $t$ and dependent variables
$\xv\left( t \right)=\left( x_{1}\left( t \right),\dots,x_{n}\left( t \right) \right)$:
\begin{equation}
\xvdd = \vec{F}_{\varepsilon}\left( \xv,\xvd \right),
\label{eq:multiscalegen}
\end{equation}
where the $\varepsilon$ subscript means that we have
dependence on a ``small'' parameter $\varepsilon$, i.e.
$\varepsilon\to0^{+}$.
From now on this condition on the parameter $\varepsilon$ will be
always assumed.
{We suppose that $\xv$ has an asymptotic expansion of the form:
\begin{equation}
\xv\left( t \right)
= \sum_{i=0}^{N+M} \varepsilon^{i} \xv^{(i)}\left( t_{0},t_{1},t_{2},\dots, t_{N} \right)
+\BigO{\epsilon^{N+M+1}},
\label{eq:asymptseries}
\end{equation}
truncated at some positive integer $N+M$, with $M\geq0$. In the right hand 
side of \eqref{eq:asymptseries} the dependence on the time variable $t$ 
appears through the so-called \emph{scales}\footnote{If $t$ is a time 
variable, the scales are the characteristic 
\emph{time scales} of $\vec{x}$. Similarly, if $t$ is 
a length variable, the scales are the characteristic 
\emph{length scales} of $\vec{x}$.} ${t_{i}=t_{i}(t,\varepsilon)}$. 
Intuitively, the scales isolate different behaviors inside 
equation \eqref{eq:multiscalegen}.
E.g. in the damped harmonic oscillator the oscillations 
and the amplitude suppression are phenomena happening on different
time scales.
The number of scales to be introduced depends on the desired  
asymptotic approximation order: the expansion is guaranteed 
to be asymptotic until 
\begin{equation}
t_{N}\left( t,\varepsilon \right) = \BigO{1}
\label{eq:scalemeaning}
\end{equation}
is satisfied.
The number of scales also sets the approximation error,
in the sense that the maximum 
discrepancy from the complete solution
\begin{equation}
\max_{t\in \left[ 0,t_{max} \right]}\abs{\xv(t)-
    \sum_{i=0}^{N+M-1} \varepsilon^{i} \xv^{(i)}\left( t_{0},t_{1},t_{2},\dots, t_{N} \right)}
\label{eq:maxfunc}
\end{equation}
is $\BigO{\varepsilon^{N+M}}$, where $t_{max}$ 
is the time such that the condition \eqref{eq:scalemeaning}
holds.

The mathematical structure of the scales is the most delicate point 
in the whole expansion method: it involves the knowledge of
the structure of the system \eqref{eq:multiscalegen}, 
and the constraint that they  must be \emph{non-decreasing} 
functions of $t$ satisfying the 
condition:
\begin{equation}
\lim_{\varepsilon\to0^+}\dfrac{t_{i+1}(t,\varepsilon)}{t_{i}(t,\varepsilon)} = 0, 
\quad \forall \, i=0,1,\dots, N-1.
\label{eq:timescalecond}
\end{equation}
Condition \eqref{eq:timescalecond} just states that the scales
are \emph{well-ordered}, i.e. phenomena happening at the scale
$t_{i+1}$ are slower than those happening on scale $t_{i}$.
In many cases, and in our paper we will do so, one can
just consider the so-called \emph{trivial time scales}:
\begin{equation}
t_{i} = \varepsilon^{i} t,
\label{eq:trivialts}
\end{equation}
which are non-decreasing linear functions of $t$ and
satisfy the condition \eqref{eq:timescalecond}.
{We note that in general $N$ has to be sufficiently high not just to give a 
longer asymptotic range of validity of the expansion, 
but also to capture the behavior of the system.}

\noindent The substitution \eqref{eq:asymptseries} can be extended
to all the derivatives of $\xv$ by differentiation or, more operatively,
by substituting 
\begin{equation}
\frac{\ud}{\ud t} = \sum_{i=0}^{N}\pdv{t_{i}}{t}\pdv{}{t_{i}},
\label{eq:asymptder}
\end{equation}
which in the case of the trivial time scales is particularly simple:
\begin{equation}
\frac{\ud}{\ud t} = \sum_{i=0}^{N}\varepsilon^{i}\pdv{}{t_{i}}.
\label{eq:asymptdertrivial}
\end{equation}
Substituting the series \eqref{eq:asymptseries} and all its 
derivatives in equation \eqref{eq:multiscalegen} using \eqref{eq:asymptder},
and eventually expanding in Taylor series with respect to $\epsilon$, 
we obtain a polynomial in $\varepsilon$ which must be identically equal to zero. 
We can then separately set to zero all the coefficients of 
$\varepsilon$-powers and obtain a system of $N+M+1$ partial differential equations. 
If the scales are correctly  chosen, the  $\varepsilon^0$-equation
will contain $\xv^{(0)}$ only and will depend just on $t_{0}$. This will
give rise to a solution depending on arbitrary functions of the remaining
scales $t_{1},\dots,t_{N}$. Substituting $\xv^{(0)}$ into the $\varepsilon^1$-equation  
we  use these arbitrary functions to prevent the birth
of the secular terms in $\xv^{(1)}$. 
Solving iteratively for the remaining $\xv^{(i)}$ one finally writes
down the $N+M$ terms of the desired expansion \eqref{eq:asymptseries}.
In the case of high order expansions ($N>1$) sometimes
the previous iterative method is not sufficient to completely specify the terms
of the asymptotic series. In these cases the strategy 
of the \emph{suppression of the order mixing} is adopted:
it consists in eliminating from the $\varepsilon^{i+1}$-equation all the contributions coming
from the arbitrary functions arising from lower orders solutions $\xv^{(i)}$, $\xv^{(i-1)}$, etc. 
This increases the accuracy of the first $i$ terms by reducing the amount of corrective terms in $\xv^{(i+1)}$ \cite{Holmes}.

\begin{remark}
Unlike the usual pertubation theory
where the object seeked is of increasing precision in
$\varepsilon$, the aim of the multiple scales method
is to derive an object of ``minimal'' precision, but valid
for a longer time scale, i.e. a true \emph{asymptotic expansion}.
If one sets in \eqref{eq:asymptseries} $M=0$, as it is usually done,
 then the output of the method is an $\varepsilon$-precision
approximate solution of the form:
\begin{equation}
    \xv = \xv^{(0)}\left( t_{0},t_{1},\dots,t_{N} \right)
    +\BigO{\varepsilon},
    \label{eq:standapp}
\end{equation}
valid until $t_{N}=\BigO{1}$.
\end{remark}

In order to disprove or suggest the maximal superintegrability
of a Hamiltonian system we will need several steps. 
Indeed, if we want to apply the multiple scales method, we will
need to introduce into our Hamilton's equations a small parameter 
$\varepsilon$ which is not, in general, naturally present.
To this end we will use a perturbative approach to equilibrium.
Moreover, we will need the following technical result:
\begin{lemma}
    Let us suppose that we are given a system of second-order differential
    equations in the form \eqref{eq:multiscalegen}, such that
    the dependence on the parameter $\varepsilon$ is analytic
    in a neighbourhood of $\varepsilon=0$.
    Assume that every bounded solution of \eqref{eq:multiscalegen}
    is periodic.
    Then every bounded perturbative series, analytic in a neighbourhood
    of $\varepsilon=0$, is periodic at any order.
    \label{lem:periodic}
\end{lemma}

\begin{proof}
    Let us assume we have a bounded analytic perturbative series
    \begin{equation}
        \xv_{P}(t)  = \sum_{k=0}^{\infty} \varepsilon^{k} \xv^{(k)}(t).
        \label{eq:xvp}
    \end{equation}
    Since  $\xv_{P}$ is bounded, then
    it is periodic by assumption. 
    This means that there exists a $T$ such that 
    $\xv\left( t+T \right) = \xv\left( t \right)$ for every $t\in\R$.
    Substituting this condition in \eqref{eq:xvp}, since the series
    is analytic in $\varepsilon$, we obtain:
    \begin{equation}
        \xv^{(k)}\left( t+T \right)=\xv^{(k)}\left( t \right),
        \quad
        \forall k \in \N,
    \end{equation}
    which implies that every order of the perturbative series
    is periodic.
    Moreover, since we assumed that the system \eqref{eq:multiscalegen}
    is analytic in a neighbourhood of $\varepsilon=0$ we have that, at every order,
    the coefficients $\xv^{(k)}$ will be approximate solutions of
    the differential system for $\varepsilon\to0^{+}$.
\end{proof}

\begin{remark}
    An asymptotic series which is bounded at every power of $\varepsilon$
    by a constant $B_{k}$, independent of $\varepsilon$,
    will satisfy the hypothesis of Lemma \ref{lem:periodic}
    provided that the sequence $B_{k}$ is well behaved.
    Indeed, if for every $k\in\N$ the estimate
    $\abs{\vec{x}^{\left( k \right)}}\leq B_{k}$,
    we have from \eqref{eq:xvp}:
    \begin{equation}
        \abs{\xv_{P}}  \leq \sum_{k=0}^{\infty} \varepsilon^{k} B_{k}.
        \label{eq:xvpabs}
    \end{equation}
    The latter series is convergent, by the Cauchy-Hadamard theorem
    \cite{WhittakerWatson1927}, if
    \begin{equation}
        R^{-1} = \limsup_{k \to \infty} |B_{k}|^{1/k} < \infty.
        \label{eq:limsup}
    \end{equation}
    We remark that in a general 
    multiple-scales expansion the condition of boundness is satisfied
    at every power of $\varepsilon$.
    In general, proving explicitly the condition \eqref{eq:limsup} can be quite complicated. However, if the original equation is analytic in a neighbourhood of $\epsilon=0$, we can assume that (except for some pathological case) such condition is satisfied.
   
\end{remark}

\noindent Then, the method  we propose can be summarized in the following steps:
\begin{enumerate}
\item We put the mechanical system under consideration 
    in \emph{Lagrangian form}, $L=L\left( \vec{q},\vec{\dot{q}} \right)$
    with $\vec{q}=\left( q_{1},\dots,q_{n} \right) \in \mathcal{U} \subset \mathbb{R}^n$ and 
    $\dot{\vec{q}}=\left( \dot {q}_{1},\dots,\dot{q}_{n} \right) \in T_{\mathbf{q}}\mathcal{U}$, where $\mathcal{U}$ is an open subset of $\mathbb{R}^n$. We remark that in general, the generalized coodinates $\vec{q}$ may belong to a Riemannian manifold $\mathcal{M}$ of dimension $n$. However, our analysis is \emph{local} therefore  we can always think to be in the appropriate chart.
    In the cases we will treat in this paper we will deal with
    \emph{natural Lagrangians}, i.e. Lagrangians of the
    form:
    \begin{equation}
        L = \half \left(\vec{\dot{q}}, A\left( \vec{q}\right)  \vec{\dot{q}}\right)
        - V\left( \vec{q} \right)
        \label{eq:natlagr}
    \end{equation}
    where $A\left( \vec{q}\right) = [A_{ij}(\mathbf{q})]$ is a symmetric,  
    positive definite $n \times n$ matrix, with $(\,,\,)$ we denoted the standard
    scalar product and $V\left( \vec{q}\right)$ is the potential.
    We prefer the Lagragian form over the Hamitonian one
    of the equations of motion because identifying and isolating
    resonances for second-order differential equations is easier.
    The coordinates $\vec{q}$ are preferably \emph{unbounded}.
    The use of unbounded coordinates, e.g. Cartesian, elliptical
    and parabolic, is preferred since it is easier to keep track
    of the secular terms.
    However, the method can be applied with the required care
    even when some \emph{bounded} coordinates are present, 
    see Subsection \ref{sss:bertrand}.
\item Search for equilibrium positions as \emph{stationary points}
    of the potential $V$ i.e. as the points such that $\nabla_{\vec{q}}V = {\mathbf 0}$,
    where  $\nabla_{\vec{q}}$ denotes the gradient with respect to
    the $\vec{q}$ coordinates.
    This will give a collection of points to test, say:
    \begin{equation}
        \Set{\vec{q}_{1},\vec{q}_{2}\dots,\vec{q}_{m}},\quad m\in\N^{*}.
        \label{eq:eqpts}
    \end{equation}
    In principle the method cannot be applied if no stationary point
    exists.
    However, in some cases, it is still possible to apply it. For a
    discussion of this extension see Subsection \ref{sss:bertrand}.
\item Determine the linearized equations perturbatively using the expansion:
    \begin{equation}
        \vec{q} = \vec{q}_{l} + \varepsilon \vec{Q}, \quad l=1,\dots,m.
        \label{eq:pertv}
    \end{equation}
    In this way we introduce the needed small parameter $\varepsilon$.
    This analysis is equivalent to the classical one \cite{Braun1993,Arnold1997},
    simply the condition $\varepsilon\to0^{+}$ replaces the ``small-norm''
    requirement usually adopted.
    The small parameter is used to linearize the equation.
    In the case of natural Lagrangians \eqref{eq:natlagr} the
    Euler-Lagrange equations are given by \cite{Lee}:
    \begin{equation}
        \sum_{j=1}^{n}A_{kj}\left( \vec{q} \right)\ddot{q}_{j}
        +\sum_{i,j=1}^{n} \Gamma_{ijk}\left( \vec{q} \right) \dot{q}_{i}\dot{q}_{j}
        +\pdv{V\left( \vec{q}\right)}{q_{k}}= 0,
        \quad
        k=1,\dots,n,
        \label{eq:natel}
    \end{equation}
    where 
    \begin{equation}
        \Gamma_{ijk}\left( \vec{q} \right) 
        \doteq \frac{1}{2} \left(\frac{\partial A_{ik}(\mathbf{q})}{\partial q_j}
        +\frac{\partial A_{jk}(\mathbf{q})}{\partial q_i}
        -\frac{\partial A_{ij}(\mathbf{q})}{\partial q_k} \right) 
        \label{eq:christoffel}
    \end{equation}
    are the contracted Christoffel symbols \cite{Dubrovin1984}.
     Inserting \eqref{eq:pertv} in the Euler-Lagrange
    equations \eqref{eq:natel} we obtain:
    \begin{equation}
        \varepsilon\sum_{j=1}^{n}A_{kj}\left( \vec{q}_{l} + \varepsilon \vec{Q} \right)\ddot{Q}_{j}
        +\varepsilon^{2}\sum_{i,j = 1}^{n}\Gamma_{ijk}\left( \vec{q}_{l} + \varepsilon \vec{Q} \right) \dot{Q}_{i}\dot{Q}_{j}
        +\pdv{V\left( \vec{q}_{l} + \varepsilon \vec{Q} \right)}{q_{k}}= 0,
        \quad
        k=1,\dots,n.
        \label{eq:nateleps}
    \end{equation}
    Expanding the system \eqref{eq:nateleps} in Taylor series
    with respect to $\varepsilon$, since $\vec{q}_{l}$
    is a stationary point, we get:
    \begin{equation}
        \varepsilon  
        \sum_{j=1}^{n}\left[A_{kj}\left( \vec{q}_{l}\right)\ddot{Q}_{j}
        +\pdv{V}{q_{k},q_{j}}\left( \vec{q}_{l}\right)Q_{j} \right]
        +\BigO{\varepsilon^{2}}= 0,
        \quad
        k=1,\dots,n.
        \label{eq:lineqgen0}
    \end{equation}
    Equation \eqref{eq:lineqgen0} can be easily written in
    vector form as
    \begin{equation}
        \varepsilon\left[
            A\left( \vec{q}_{l} \right)\vec{\ddot{Q}} 
        + \mathrm{H}\left( V \right)\left( \vec{q}_{l} \right) \vec{Q}
        \right]
        +\BigO{\varepsilon^{2}}
        = \mathbf{0},
        \label{eq:lineqgen}
    \end{equation}
    denoting with $\mathrm{H}\left( V \right)\left( \vec{q}_{l} \right)$
    the Hessian matrix of $V$ evaluated in $\vec{q}_{l}$.
    We just obtained through this perturbative approach the usual linearized
    Euler-Lagrange equations \cite{Arnold1997}.
\item  Given the linearized equations \eqref{eq:lineqgen},
    it is well known that their integration 
    can be reduced to a problem in linear algebra which consists
    in finding the eigenvalues of the symmetric matrix
    $\mathrm{H}\left( V \right)\left( \vec{q}_{l} \right)$ with respect to
    the scalar product induced by the symmetric and positive definite
    matrix $A\left( \vec{q}_{l} \right)$.
    Practically, this can be done solving the characteristic equation:
    \begin{equation}
        \det \left[ \mathrm{H}\left( V \right)\left( \vec{q}_{l} \right)
        -\omega^{2} A\left( \vec{q}_{l} \right) \right]=0,
        \label{eq:chareq}
    \end{equation}
    with respect to $\omega$.
    The solution of equation \eqref{eq:chareq} 
    yields $n$ solutions for $\omega^{2}$.
    If every possible value of $\omega^{2}$ is positive, then we say
    that the equilibrium point is \emph{stable}.
    If all the possible values of $\omega^{2}$ are non-negative, then
    we say that the equilibrium point is \emph{neutral}.
    If at least one of the possible values of $\omega^{2}$ is negative, we say 
    that the equilibrium is \emph{unstable}.
    As we will discuss, stable equilibrium points give rise to bounded
    orbits, and this is the reason why we will be interested in this kind of stationary points.
    We will assume that in the set of points \eqref{eq:eqpts} there
    exists \emph{at least} one stable point, otherwise the algorithm
    is not applicable.
    Therefore, we will continue our discussion assuming that all the possible values 
    of $\omega^{2}$ are positive.  Then, we can define the vector
    $\boldsymbol{\omega}=\left( \omega_{1},\dots,\omega_{n} \right)$
    of the $n$ positive, possibly equal, solutions of \eqref{eq:chareq}.
    We call the constants $\omega_{i}$ the \emph{fundamental frequencies}. Physically, this procedure amounts to make the \emph{ansatz}:
    \begin{equation}
    \vec{Q} = \vec{S} \cos\left( \omega t + \varphi \right),
    \label{eq:ansatz}
    \end{equation}
    which inserted in \eqref{eq:lineqgen} is solution if and only if \eqref{eq:chareq} is satisfied. Since the matrices $\mathrm{H}\left( V \right)\left( \vec{q}_{l} \right)$
    and $A\left( \vec{q}_{l} \right)$ are symmetric,
    to the $n$ fundamental frequencies correspond $n$ independent 
    orthogonal eigenvectors
    $\vec{S}^{(1)},\dots,\vec{S}^{(n)}$, i.e. the solutions of:
    \begin{equation}
        \left[ \mathrm{H}\left( V \right)\left( \vec{q}_{l} \right)
        -\omega^{2}_{i} A\left( \vec{q}_{l} \right) \right] \vec{S}^{(i)}=\mathbf{0}.
        \label{eq:eigenvectors}
    \end{equation}
    The general solution of the linearized Euler-Lagrange equation is
    then given by:
    \begin{equation}
        \vec{Q} = \sum_{i=1}^{n} \vec{S}^{(i)} N_{i}\left( t \right),
        \label{eq:genlinel}
    \end{equation}
    where:
    \begin{equation}
        N_{i} (t)= R_{i} \cos\left( \omega_{i}t+\varphi_{i} \right) \,, \quad i=1, \dots, n,
        \label{eq:normsol}
    \end{equation}
    with $R_{i}$ and $\varphi_{i}$ constants.
    Then, performing the linear transformation  $\vec{Q}=S\vec{N}$
    with $S$ given by:
    \begin{equation}
        S = \left[ \vec{S}^{(1)},\dots,\vec{S}^{(n)} \right] ,
        \label{eq:Smatrix}
    \end{equation}
    we have that the system \eqref{eq:lineqgen} reduces to
    \begin{equation}
        \vec{\ddot{N}} + \Lambda_{S} \vec{N}=\mathbf{0},
        \quad 
        \Lambda_{S} 
        = 
        S^{T} A\left( \vec{q}_{l} \right)^{-1}\mathrm{H}\left( V \right)\left( \vec{q}_{l} \right) S 
        = 
        \diag\left( \omega_{1}^{2},\dots,\omega_{n}^{2} \right).
        \label{eq:lineqgen3}
    \end{equation}
    Thus, in the coordinates $\vec{N}$, the system acts
    as $n$ one-dimensional systems whose solution is given by \eqref{eq:normsol}.
    The coordinates $\vec{N}$ are called the \emph{normal coordinates}.
%
%
    This means that the solution of the system at order $\varepsilon$ is
    bounded.
    We  introduce the \emph{frequencies ratio matrix}:
    \begin{equation}
        \Delta^{(0)}\left( \vec{q}_{l} \right) 
        \doteq \left[ \frac{\omega_{i}^{(0)}}{\omega_{j}^{(0)}} \right]_{i,j=1,\dots,n}.
        \label{eq:freqratmat}
    \end{equation}
    If the entries of the matrix $\Delta^{(0)}\left( \vec{q}_{l} \right)$ \eqref{eq:freqratmat} 
    are \emph{rational} then we have found an approximate bounded periodic closed orbit.
    This implies that the system under scrutiny \emph{can be} maximally superintegrable. 
    If there exists \emph{at least} a stable equilibrium point $\vec{q}_{l^{*}}$,
    such that the matrix $\Delta^{(0)}\left( \vec{q}_{l^{*}} \right)$ 
    \eqref{eq:freqratmat} possesses \emph{at least} 
    an irrational entry then, as a consequence of Corollary \ref{cor:nonperiodic}
    and Lemma \ref{lem:periodic}, the system cannot be maximally superintegrable.
    In the latter case the algorithm terminates here with a negative answer.
\item If all the stable equilibrium points in \eqref{eq:eqpts} give
    rise to closed periodic orbits, we have to check if this conditions
    is preserved on longer time scales.
    To this end we return to the Euler-Lagrange equations,
    say of the form of \eqref{eq:natel}, but we  now assume
    that $\vec{Q}=S\vec{N}$ where $\vec{N}$ is given by a multiple-scale expansion:
    \begin{equation}
        \vec{N} = 
        \sum_{i=0}^{N+M} \varepsilon^{i} \vec{N}^{(i)}\left( t_{0},t_{1},t_{2},\dots, t_{N} \right)
        +\BigO{\epsilon^{N+M+1}},
        \label{eq:vmult}
    \end{equation}
    where as time scales we use the trivial ones \eqref{eq:trivialts}.
    This asymptotic expansion will give rise to \emph{higher order corrections}
    to the fundamental frequencies $\boldsymbol{\omega}$:
    \begin{equation}
        \Omega_{j}^{(N)} = \sum_{i=0}^{N} \varepsilon^{i}\omega_{j}^{(i)},
        \quad
        j =1,\dots,n.
        \label{eq:omegajn}
    \end{equation}
    We introduce the \emph{higher order frequencies ratio matrix}:
    \begin{equation}
        \Delta^{(N)}\left( \vec{q}_{l} \right) 
        \doteq \left[ \frac{\Omega_{i}^{(N)}}{\Omega_{j}^{(N)}} \right]_{i,j=1,\dots,n}.
        \label{eq:freqratmatn}
    \end{equation}
    Again, if the entries of the matrix $\Delta^{(N)}\left( \vec{q}_{l} \right)$ 
    are \emph{rational}, we have found an approximate bounded 
    periodic closed orbit of order $N$.
    This can suggest maximal superintegrability.
    If on the contrary we are able to prove that there exists 
    \emph{at least} a point $\vec{q}_{l^{*}}$ and an integer $N^{*}>1$,
    such that one of the entries of the matrix $\Delta^{(N^{*})}\left( \vec{q}_{l^*} \right)$ 
    \eqref{eq:freqratmatn} is not rational then, from Corollary \ref{cor:nonperiodic}
    and Lemma \ref{lem:periodic},
    we can conclude that the system is not maximally superintegrable.
    \label{it:mult}
\end{enumerate}

\begin{remark}
    We observe that in the plane case, i.e. when $n=2$, it is not
    necessary to consider the ``full'' matrices \eqref{eq:freqratmat}
    and \eqref{eq:freqratmatn}, but only the ratios
    \begin{equation}
        \Delta^{(0)}\left( \vec{q}_{l} \right) 
        = \frac{\omega_{1}^{(0)}}{\omega_{2}^{(0)}}
        \quad
        \text{and}
        \quad
        \Delta^{(N)}\left( \vec{q}_{l} \right) 
        = \frac{\Omega_{1}^{(N)}}{\Omega_{2}^{(N)}}
        \label{eq:freqrat2dim}
    \end{equation}
    will be relevant.
\end{remark}

\begin{remark}
At the present stage, in order to show that a system
is not maximally superintegrable, it was sufficient to take
an expansion \eqref{eq:vmult} such that $N=2$ and $M=0$.
It is not known if there exist non-maximally superintegrable
systems requiring higher order expansions.
\end{remark}

In the next Section we present some examples of the application
of the method we just outlined.

\section{Examples of application of the method}

\label{sec:examples}

In this Section we present the practical application of
the method we explained in Section \ref{sec:method}.
In particular we provide five examples which are aimed to underline
the importance and the possibilities of the presented procedure.

In Subsections \ref{sss:henonheiles} and \ref{sss:anisco} 
we present two simple examples, namely the \emph{generalized H\'enon-Heiles system} 
\cite{HenonHeiles1964,Fordy1991}
and the \emph{anisotropic caged oscillator} \cite{Evans2008}.
We use these two examples to discuss the range of the results which can
be obtained in the framework of our method, and to confront them with other
algorithmic procedures aimed to find integrable systems.
In Subsection \ref{sec:ttw} we show that the superintegrability of
the \emph{Tremblay-Turbiner-Winternitz} (TTW) system \cite{TTW2009}, 
can be inferred using the algorithm of Section \ref{sec:method}.
In Subsection \ref{sec:drach} we present a new result about the so-called 
\emph{Drach} system \cite{Drach1935}. 
Applying the method presented in Section \ref{sec:method} we
show that the Drach system is not, for general values of the parameters,
superintegrable.
This result answers a comment made by the authors of  \cite{PostWinternitz2011}
where a particular case of this model,
which we will call the \emph{Drach-Post-Winternitz} (DPW) system,
was considered and it was shown to be superintegrable. 
Finally, in Subsection \ref{sss:bertrand}, we show that the method of 
Section \ref{sec:method} can be applied as a sieve test for (maximally)
superintegrable systems. Indeed, we prove in this framework
  \emph{Bertrand's theorem}  \cite{Bertrand1873}, which characterizes all the possible
 central potentials in the plane with bounded and closed orbits.

\subsection{The generalized H\'enon-Heiles system}

\label{sss:henonheiles}

As a first application of the method we discuss the so-called \emph{generalized H\'enon-Heiles Hamiltonian}:
\begin{equation}
    H_{\text{HH}} = 
    \half \left( p_{1}^{2}+  p_{2}^{2} \right)
    +\frac{\omega_{1}^{2}}{2}q_{1}^{2}
    +\frac{\omega_{2}^{2}}{2} q_{2}^{2}
    +\alpha  \left( q_{1}^{2}q_{2} + \beta q_{2}^{3} \right), 
    \label{eq:HHH}
\end{equation}
where $\omega_{1}$, $\omega_{2} \in \mathbb{R}^+$ and $\alpha$, $\beta$
are real parameters.
This system is called generalized H\'enon-Heiles since
it is a generalization of the H\'enon-Heiles system
which arises for $\omega_{1}=\omega_{2}=1$ and $\beta=-1/3$.
The H\'enon-Heiles system was introduced in \cite{HenonHeiles1964} 
to model the dynamics of a Newtonian axially-symmetric galactic system and it is usually regarded as a prototype
of Hamiltonian system which exhibits chaotic behavior 
\cite{Tabor1989,Gutzwiller1990,BoccalettiPucacco2004}.
Despite these facts several \emph{integrable} subcases 
of \eqref{eq:HHH} are known 
\cite{ChangTaborWeiss1982,BountisSegurVivaldi1982,GrammaticosDorizziPadjen1982}:
\begin{enumerate}
    \item The Sawada-Kotera case: $\beta=1/3$ and $\omega_{1}=\omega_{2}$.
    \item The KdV case: $\beta=2$ and $\omega_{1}$, $\omega_{2}$ arbitrary.
    \item The Kaup-Kupershmidt case: $\beta=16/3$ and $\omega_{2}=4\omega_{1}$.
\end{enumerate}
We used the terminology of \cite{Fordy1991}, since these three integrable
cases correspond to the stationary flows of three integrable fifth-order
polynomial nonlinear evolution equations.
The Lagrangian corresponding to the Hamiltonian \eqref{eq:HHH} is:
\begin{equation}
    L_{\text{HH}} = 
    \half \left( \dot{q}_{1}^{2}+  \dot{q}_{2}^{2} \right)
    -\frac{\omega_{1}^{2}}{2}q_{1}^{2}
    -\frac{\omega_{2}^{2}}{2} q_{2}^{2}
    -\alpha  \left( q_{1}^{2}q_{2} + \beta q_{2}^{3} \right) 
    \label{eq:HHL}
\end{equation}
and its Euler-Lagrange equations are given by:
\begin{subequations}
    \begin{align}
        \ddot{q}_{1}+\omega_{1}^{2}{  q_{1}} +2 \alpha {  q_{1}} {  q_{2}} =0,
        \label{eq:HHeqa}
        \\
        \ddot{q}_{2}+ \omega_{2}^{2}{  q_{2}} 
        +\alpha  \left( q_{1}^{2} +3 \beta  q_{2}^{2} \right) =0.
        \label{eq:HHeqb}
    \end{align}
    \label{eq:HHeqs}
\end{subequations}
It is easy to verify that $\left( q_{1},q_{2} \right)= \left( 0,0 \right)$ 
is a stable equilibrium position for the system \eqref{eq:HHeqs}. 
From equation \eqref{eq:pertv} we introduce:
\begin{equation}
    q_{1} = \varepsilon Q_{1},\quad q_{2} = \varepsilon Q_{2}.
    \label{eq:HHv}
\end{equation}
The linearized system is given by:
\begin{subequations}
    \begin{align}
        \ddot{Q}_{1}+\omega_{1}^{2}{Q_{1}} =0,
        \label{eq:HHeqlina}
        \\
        \ddot{Q}_{2}+ \omega_{2}^{2}{Q_{2}} =0,
        \label{eq:HHeqlinb}
    \end{align}
    \label{eq:HHeqslin}
\end{subequations}
and it is already in normal form.
The fundamental frequencies are 
$\boldsymbol{\omega}=\left( \omega_{1},\omega_{2} \right)$.
Therefore, it seems at this stage that every $\omega_{1},\omega_{2}\in\R^{+}$
with rational ratio can give rise to a maximally superintegrable
system.
To check what happens at higher order we perform the multiple-scale expansion 
of $Q_{1}$ and $Q_{2}$, with the three trivial time scales $t_{i}$ \eqref{eq:trivialts} 
with $i=0,1,2$:
\begin{equation}
    Q_{k} = Q_{k}^{(0)}\left(t_{0},t_{1},t_{2} \right)
    +\varepsilon Q_{k}^{(1)}\left( t_{0},t_{1},t_{2} \right)
    +\varepsilon^{2} Q_{k}^{(2)}\left( t_{0},t_{1},t_{2} \right)
    + \BigO{\varepsilon^{3}},
    \quad k=1,2.
    \label{eq:HHvpert}
\end{equation}
We substitute \eqref{eq:HHvpert} into \eqref{eq:HHv}
and then into \eqref{eq:HHeqs}. Solving the obtained equations
we get the following solution:
\begin{subequations}
    \begin{align}
        Q_{1} &= R_{1} 
        \cos\left[ 
            \left( \omega_{1}+\omega_{1}^{(2)}\varepsilon^{2} \right)t
        +\varphi_{1}\right]
        +\BigO{\varepsilon},
        \label{eq:HHv1}
        \\
        Q_{2} &= R_{2} 
        \cos\left[ 
            \left( \omega_{2}+\omega_{2}^{(2)}\varepsilon^{2} \right)t
        +\varphi_{2}\right]
        +\BigO{\varepsilon},
        \label{eq:HHv2}
    \end{align}
    \label{eq:HHvi}
\end{subequations}
where $R_{1}$, $R_{2}$, $\varphi_{1}$, $\varphi_{2}$ are integration
constants and (provided  $\omega_2 \neq 2 \omega_1$):
\begin{subequations}
    \begin{align}
        \omega_{1}^{(2)} & =
        -\frac{1}{4} {\frac {{\alpha}^{2} 
        \left( 8 R_{1}^{2}\omega_{1}^{2}-3  R_{1}^{2}\omega_{2}^{2}
        +4 R_{2}^{2}\omega_{2}^{2} +24 \beta R_{2}^{2}\omega_{1}^{2}
        -6 \beta R_{2}^{2}\omega_{2}^{2} \right)}{%
            {  \omega_{1}} \omega_{2}^{2} \left( 4 \omega_{1}^{2}-\omega_{2}^{2} \right) }}, 
        \label{eq:HHom12}
        \\
        \omega_{2}^{(2)} &=
        -\frac{1}{4} {\frac {{\alpha}^{2} 
        \left( 4 R_{1}^{2}\omega_{2}^{2}+60 {\beta}^{2}R_{2}^{2}\omega_{1}^{2}
        -15 {\beta}^{2}  R_{2}^{2}\omega_{2}^{2}+24 \beta R_{1}^{2}\omega_{1}^{2}
        -6 \beta \omega_{2}^{2}R_{1}^{2} \right) }{\omega_{2}^{3} 
        \left( 4 \omega_{1}^{2}-\omega_{2}^{2} \right) }}.
        \label{eq:HHom22}
    \end{align}
    \label{eq:HHom2}
\end{subequations}
Clearly, the ratio:
\begin{equation}
    \Delta^{(2)} \doteq  \frac{\omega_{1}+\omega_{1}^{(2)}\varepsilon^{2}}{\omega_{2}+\omega_{2}^{(2)}\varepsilon^{2}}, 
    \label{eq:HHratio}
\end{equation}
is no longer a rational number independently from the values of
the initial conditions.
This means that, for arbitrary values of the parameters, the generalized H\'enon-Heiles system \eqref{eq:HHH}
cannot be (maximally) superintegrable.
We can ask ourselves if there are some superintegrable subcases
of the generalized H\'enon-Heiles system \eqref{eq:HHH} 
trying to annihilate the terms depending on the initial conditions
in \eqref{eq:HHratio}.
The procedure is the following: we expand \eqref{eq:HHratio} in Taylor
series with respect to $\varepsilon$ and then try to annihilate
the terms depending on the initial conditions.
To this end, as a first step, it is sufficient to look for an
expansion up to $\varepsilon^{2}$:
\begin{equation}
    \Delta^{(2)} = \frac{\omega_{1}}{\omega_{2}}
    +
    \frac{1}{4}\frac{\alpha^{2}}{\omega_1\omega_2^5\left(4\omega_1^{2}-\omega_2^{2}\right)}
    \left[ 
        \begin{aligned}
            &\phantom{+}\left(24\beta \omega_1^4 -4 \omega_2^2 \omega_1^2
            +3 \omega_2^4-6  \beta \omega_1^2 \omega_2^2\right) R_{1}^2
            \\
            &+\left(60 \beta^2 \omega_1^4-24 \beta \omega_1^2  \omega_2^2
            +6 \beta \omega_2^4 -4 \omega_2^4-15 \beta^2 \omega_1^2 \omega_2^2\right) R_{2}^2
        \end{aligned}
    \right]
    \epsilon^2+\BigO{\varepsilon^{3}}.
    \label{eq:HHrationexp}
\end{equation}
In \eqref{eq:HHrationexp} the initial conditions are represented
by $R_{1}$ and $R_{2}$. 
Therefore, since \eqref{eq:HHrationexp} is a polynomial in
$R_{1}$ and $R_{2}$, to make it independent of the initial conditions
we can take just coefficients with
respect to $R_{1}$ and $R_{2}$ and annihilate them separately.
Doing so we find that
the parameter of a possible maximally superintegrable H\'enon-Heiles 
system should satisfy the following equations:
\begin{subequations}
    \begin{align}
        \alpha^{2}\left(
        24\beta \omega_1^4 -4 \omega_2^2 \omega_1^2
        +3 \omega_2^4-6 \beta \omega_1^2  \omega_2^2\right) &=0,
        \label{eq:supHHa}
        \\
        \alpha^{2}\left(60 \beta^2 \omega_1^4-24 \beta \omega_1^2  \omega_2^2
        +6  \beta \omega_2^4-4 \omega_2^4-15 \beta^2 \omega_1^2 \omega_2^2\right)&=0.
        \label{eq:supHHb}
    \end{align}
    \label{eq:supHH}
\end{subequations}
Discarding the trivial solution $\alpha=0$, which would imply that
the system is \emph{linear}, we obtain the following values for the 
parameters\footnote{We discard two complex-conjugate
solutions which do not satisfy the requirements on the parameter $\omega_{2}$.}:

    \begin{equation}
        \omega_{2} = \pm  \frac{2}{3}\omega_{1} \sqrt{6\sqrt{11}-15} \, , \quad 
          \beta = \frac{2}{15}
        (7 \sqrt{11}-27).
        \label{eq:betaHH}
    \end{equation}
Unfortunately, this result is incompatible with the zero-th order
condition on $\omega_{1}$ and $\omega_{2}$, since the ratio 
$\omega_{1}/\omega_{2}$ is not a rational
number.
We can therefore conclude that the generalized H\'enon-Heiles system
is not (maximally) superintegrable and no (maximally) superintegrable subcases
exist.
This result is consistent with the literature. 
On the other hand, we see that from this approach we obtained no information 
about the various integrable cases discussed above. 
The integrable subcases are indistinguishable from the chaotic ones.

To give the feeling of the form of the trajectories
of the generalized H\'enon-Heiles system \eqref{eq:HHH} 
we show some examples in Figure \ref{fig:henonheiles}, where
it is possible to appreciate the fact that the trajectories are
not closed.
\begin{figure}[h!bt]
\centering
\subfloat[][Sawada-Kotera case with $\omega_{1}=\omega_{2}=1$.]{%
    \includegraphics[width=0.5\textwidth]{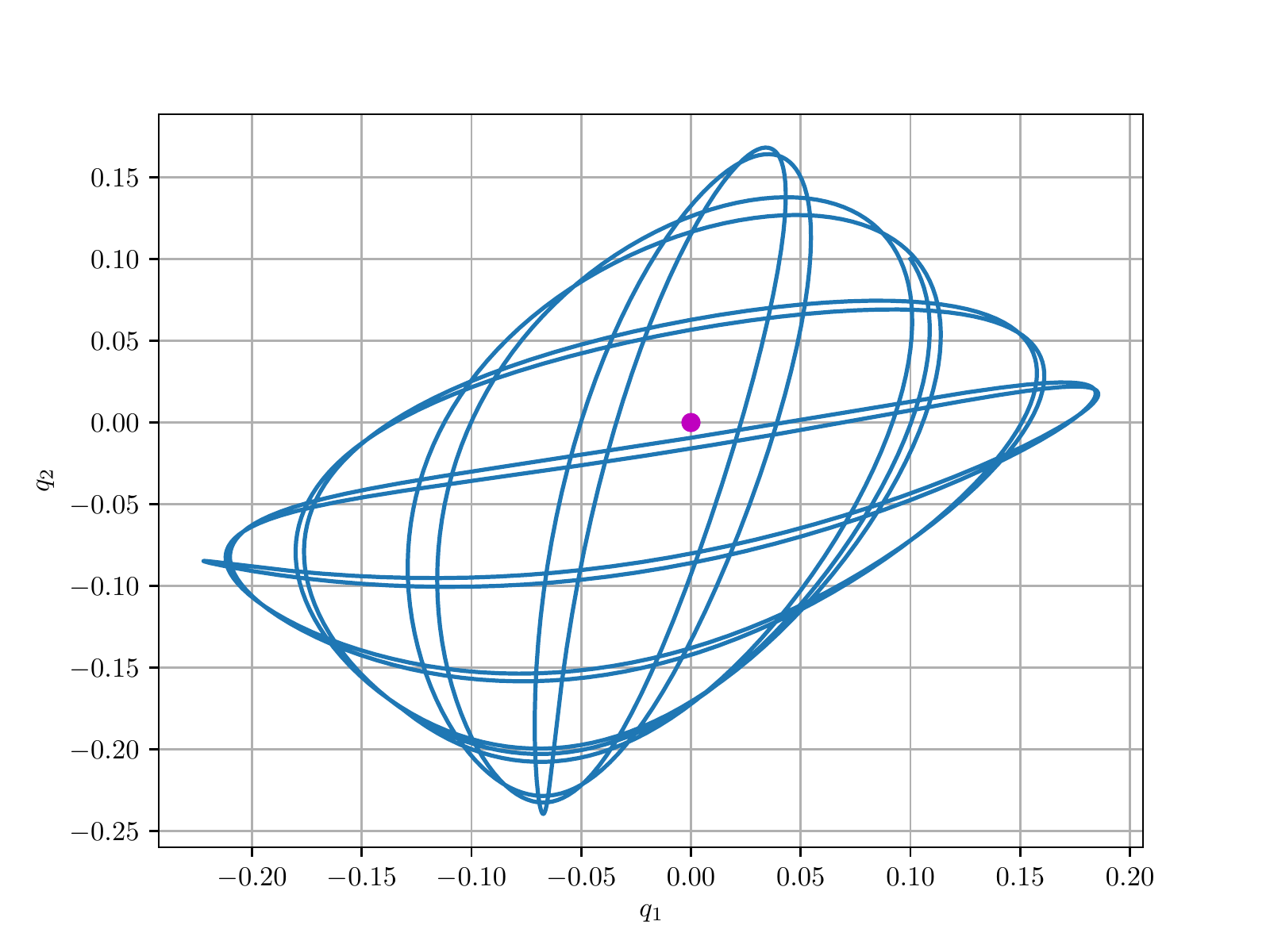}
    }
    \subfloat[][KdV case with $\omega_{1}=1$, $\omega_{2}=5$.]{
    \includegraphics[width=0.5\textwidth]{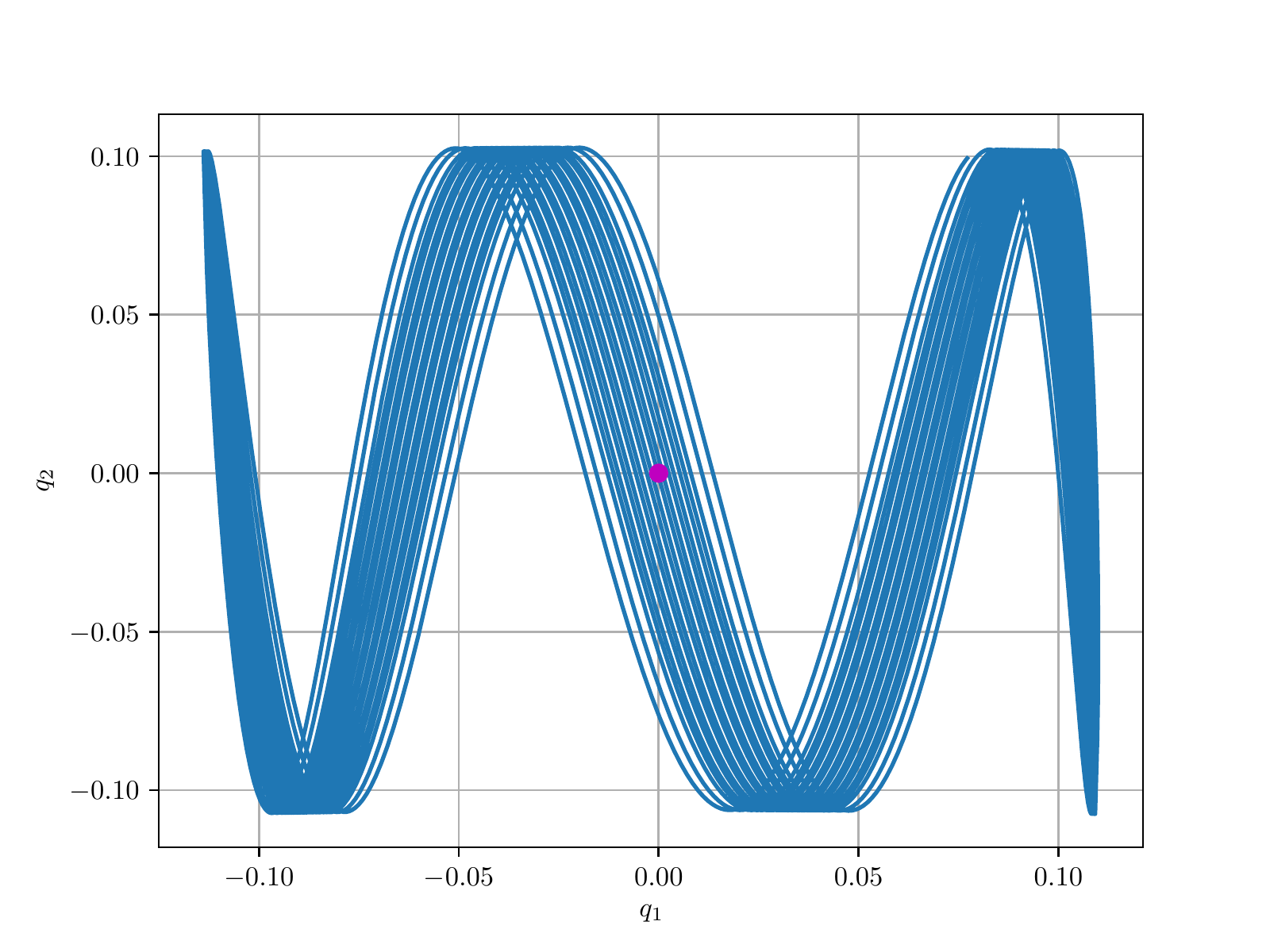}
}
\\
\subfloat[][Kaup-Kupershmidt case with $\omega_{1}=1$, $\omega_{2}=4$.]{
    \includegraphics[width=0.5\textwidth]{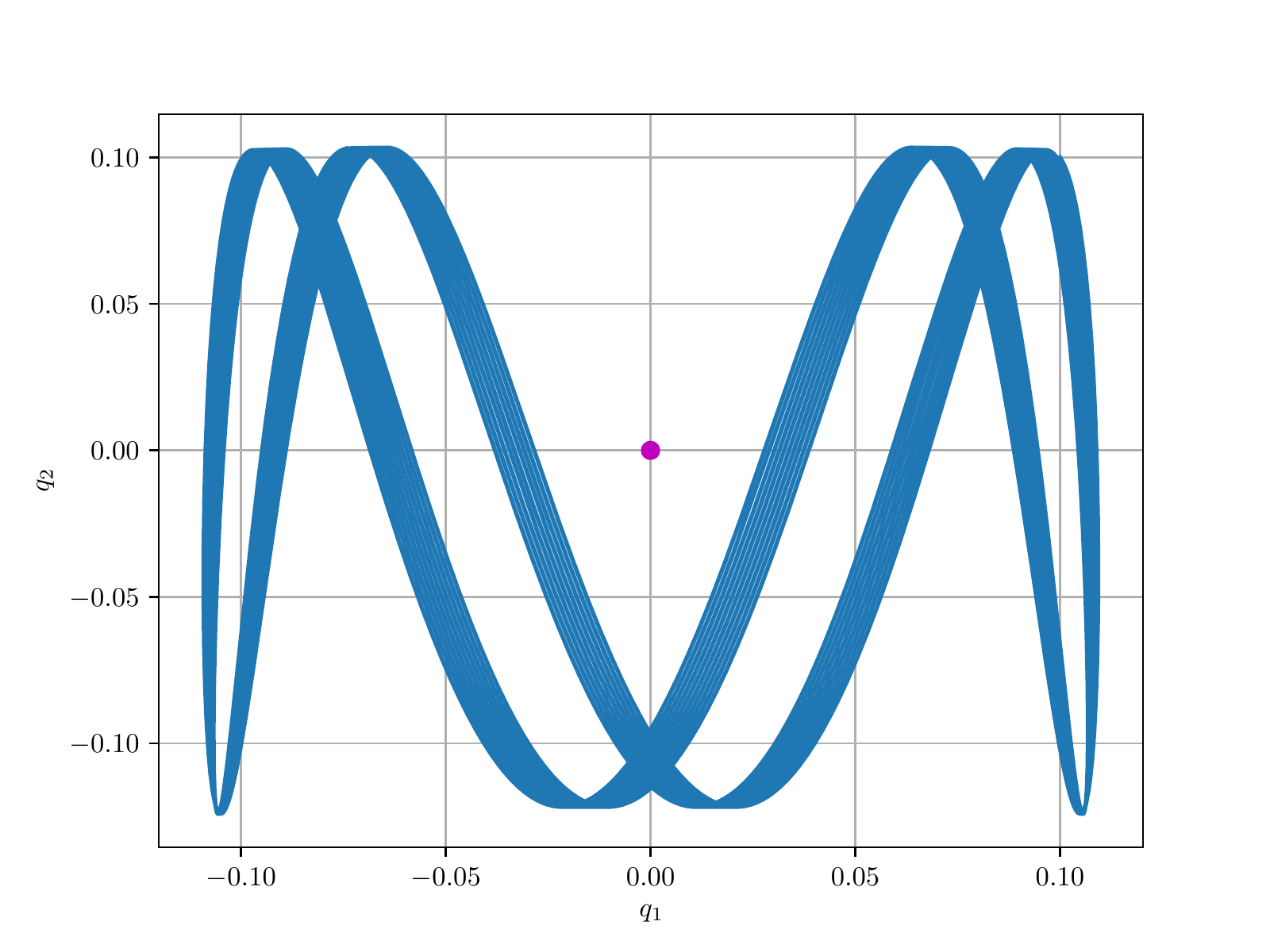}
}
\subfloat[][H\'enon-Heiles case with $\omega_{1}=\omega_{2}=1$.]{
    \includegraphics[width=0.5\textwidth]{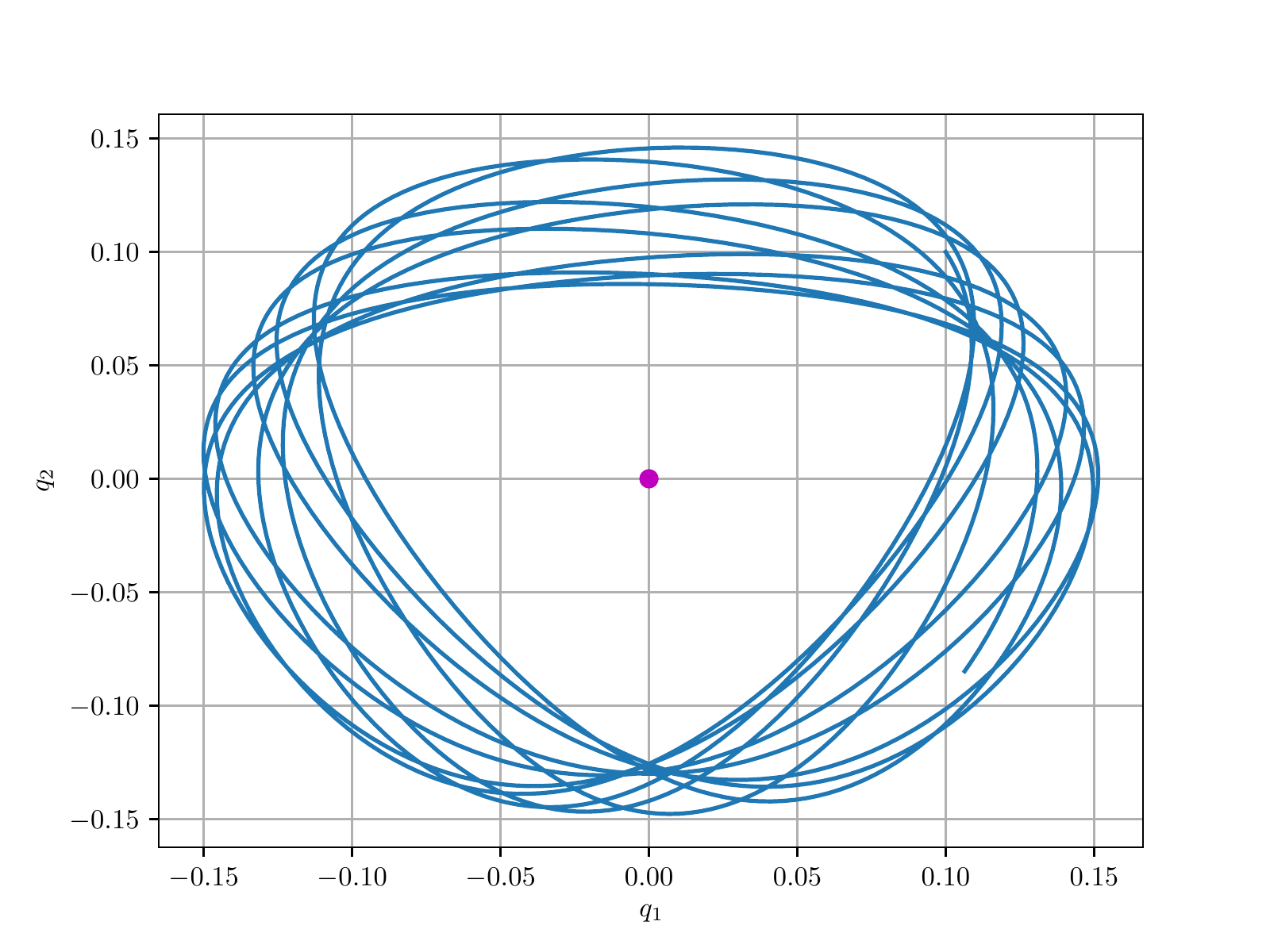}
}
\caption{Trajectories of the generalized H\'enon-Heiles Hamiltonian
    \eqref{eq:HHH} in the various integrable cases and in the chaotic
    case. We used the \texttt{odeint} integrator from \texttt{scipy} \cite{scipy}
    with a regular mesh of $N=2^{25}$ points in the interval 
    $[0,10\times2\pi/\omega_{1}]$ and $\alpha=2.1$.
    The magenta dot is the stable equilibrium point $(0,0)$.}
\label{fig:henonheiles}
\end{figure}

\begin{remark}
    We underline the fact that our method,
    as a consequence of Corollary \ref{cor:nonperiodic}, is not
    able to distinguish between integrable and non-integrable cases.
    The integrable cases of the generalized H\'enon-Heiles system 
    \eqref{eq:HHH} were found in 
    \cite{ChangTaborWeiss1982,BountisSegurVivaldi1982,GrammaticosDorizziPadjen1982}
    by means of the so-called Painlev\'e test, i.e requiring
    that the solutions of the equations of motion possess the
    Painlev\'e property \cite{Painleve1900,Painleve1902,Picard1889,Gambier1910,Fuchs1906}.
    It is said that a differential equation possesses the Painlev\'{e} property when 
    the only movable singularities are poles (for a modern exposition
    see \cite{painleveproperty1999,Hone2009}).
    Therefore, one may think that our method is stronger
    than the Painlev\'e test which instead can also yield integrability
    and not maximal superintegrability.
    However, it is easy to see that statements made using the
    Painlev\'e test and our approach are usually not comparable, i.e.
    when one method fails, the other can succeed and viceversa.
    To make clear this point, let us discuss the example of the \emph{anisotropic caged oscillator}.
\end{remark}

\subsection{The anisotropic caged oscillator}

\label{sss:anisco}

Let us consider the Hamiltonian:
\begin{equation}
    H_{\text{c.o.}} = \half\left( p_{1}^{2}+p_{2}^{2}+p_{3}^{2} \right)
    + \frac{\omega^{2}}{2} 
    \left( l_{1}^{2}q_{1}^{2}+l_{2}^{2}q_{2}^{2}+l_{3}^{2}q_{3}^{2} \right)
    +\frac{\alpha_{1}}{2q_{1}^{2}}+\frac{\alpha_{2}}{2q_{2}^{2}}+\frac{\alpha_{3}}{2q_{3}^{2}},
    \label{eq:cagedoscillator}
\end{equation}
where $l_k$ are integers, $\omega$ a positive real number and
$\alpha_k$ real constants. This system, which was proved
to be maximally superintegrable in \cite{Evans2008}, represents
a generalization of the so-called Smorodinski-Winternitz oscillator
\cite{Fris1965,Evans1990}, corresponding to the case $l_1=l_2=l_3=1$.
The Lagrangian associated to \eqref{eq:cagedoscillator} is
given by:
\begin{equation}
    L_{\text{c.o.}} = \half\left( \dot{q}_{1}^{2}+\dot{q}_{2}^{2}+\dot{q}_{3}^{2} \right)
    - \frac{\omega^{2}}{2} 
    \left( l_{1}^{2}q_{1}^{2}+l_{2}^{2}q_{2}^{2}+l_{3}^{2}q_{3}^{2} \right)
    -\frac{\alpha_{1}}{2q_{1}^{2}}-\frac{\alpha_{2}}{2q_{2}^{2}}-\frac{\alpha_{3}}{2q_{3}^{2}},
    \label{eq:cagedoscillatorlagr}
\end{equation}
and the Euler-Lagrange equations corresponding to \eqref{eq:cagedoscillatorlagr}
are:
\begin{equation}        
    \ddot{q}_{k}+ l_{k}^{2}\omega^2 q_{k} = \frac{\alpha_{k}}{q_{k}^3},
    \quad k=1,2,3.
    \label{eq:cagedoscillatorel}
\end{equation}
If $\alpha_{k}$ is positive for $k=1,2,3$ we have the following
equilibrium positions:
\begin{equation}
    q_{k}^{(\pm)} =\pm \left( \frac{\alpha_{k}}{l_{k}^{2}\omega^{2}} \right)^{1/4},
    \quad k=1,2,3.
    \label{eq:cagedoscillatorequil}
\end{equation}
From equation \eqref{eq:pertv} we introduce:
\begin{equation}
    q_{k} = q_{k}^{(\pm)}+\varepsilon Q_{k},
    \quad
    k=1,2,3.
    \label{eq:cagedoscillatorv}
\end{equation}
Then, the linearized system is given by:
\begin{equation}
    \ddot{Q}_{k}+4l_{k}^{2}\omega^{2}{  Q_{k}} =0,
    \quad
    k=1,2,3.
    \label{eq:cagedoscillatoreqlin1}
\end{equation}
The system is already in normal form and
the fundamental frequencies are
$\boldsymbol{\omega} = \left( 2l_{1}\omega,2l_{2}\omega,2l_{3}\omega \right)$.
The ratios of the fundamental frequencies are rational as long as $l_{k}$ are integers.
To check if this property is preserved at higher order we perform 
the multiple-scale expansion of the $Q_{k}$, where $k=1,2,3 \,$, using 
three trivial time scales $t_{i}$ \eqref{eq:trivialts} with $i=0,1,2$:
\begin{equation}
    Q_{k} = Q_{k}^{(0)}\left( t_{0},t_{1},t_{2} \right)
    +\varepsilon Q_{k}^{(1)}\left( t_{0},t_{1},t_{2} \right)
    +\varepsilon^{2} Q_{k}^{(2)}\left( t_{0},t_{1},t_{2} \right)
    + \BigO{\varepsilon^{3}},
    \quad k=1,2,3.
    \label{eq:cagedvpert}
\end{equation}
We substitute \eqref{eq:cagedvpert} into \eqref{eq:cagedoscillatorv}
and then into \eqref{eq:cagedoscillatorel}. Solving the obtained equations
we get the following solution:
\begin{equation}
    Q_{k} = R_{k} \cos\left( 2l_{k}\omega+\varphi_{k} \right),
    \quad
    k=1,2,3,
    \label{eq:cagedoscillatorvi}
\end{equation}
where $R_{k}$ and $\varphi_{k}$ are integration constants.
We have then that the first order condition is unaltered,
and our approach gives affirmative output, suggesting
maximal superintegrability.
On the other hand the system \eqref{eq:cagedoscillatorel}
does not pass the Painlev\'e test\footnote{Following
    \cite{Hone2009}, since the system \eqref{eq:cagedoscillatorel} is
rational in $\vec{q}$, the Painlev\'e test is applicable.}.
In fact, let us assume to have
a movable singular point $t_{0}$ with the following
behavior in a neighborhood of $t_{0}$:
\begin{equation}
    q_{k} = c^{(k)}_{0} (t-t_{0})^{\mu_{k}}
    + \BigO{\left|t-t_{0}\right|^{\mu_{k}+1}},
    \quad
    k=1,2,3.
    \label{eq:fundterm}
\end{equation}
Inserting this ansatz into the equations of motion
\eqref{eq:cagedoscillatorel},
the possible balances which yields the value of
the $\mu_{k}$, are $\mu_{k}=0,1/2$
for $k=1,2,3$.
This means that the solutions of
\eqref{eq:cagedoscillatorel} is either expressible in
Taylor series, then no singularities occur, or it possesses
a branch cut, i.e. the behavior is algebraic. 
In both cases it does not possess movable poles.
Moreover, it is easy to show that the series expansion of
the solutions of \eqref{eq:cagedoscillatorel} 
when $\mu_{k}=1/2$ is of the form:
\begin{equation}
    q_{k} = \sum_{n=0}^{\infty} c_{n}^{(k)} \left( t-t_{0} \right)^{n+1/2},
    \label{eq:cagedexp}
\end{equation}
i.e. is a \emph{Puiseux series} \cite{Puiseux1850,Puiseux1851,Shafarevich1994}.
Therefore, using the Painlev\'e analysis it is not possible to infer
the maximal superintegrability of the caged anisotropic oscillator
\eqref{eq:cagedoscillator}, whereas with our method it is.
We underline that it is known that there exists many \emph{integrable}
systems which do not possess the Painlev\'e property, but only
the so-called \emph{weak Painlev\'e property}, which allows
the appearance of branch cuts as in \eqref{eq:cagedexp}.
Examples of this kind of systems can be found in 
\cite{RamaniDorizziGrammaticos1982,DorizziGrammaticosRamani1983,
AbendaFedorov2000,AbendaFedorov2001}.

A complete discussion and some graphs of the trajectories of 
the anisotropic caged oscillator \eqref{eq:cagedoscillator}
can be found in \cite{Evans2008}.

Moreover, we observe that the generalization of this example
to the $n$ dimensional case, i.e. to:
\begin{equation}
    H_{\text{c.o.}}^{(n)} = \sum_{k=1}^{n}
    \left(\frac{p_{k}^{2}}{2}+ \frac{\omega^{2}}{2}l_{k}^{2}q_{k}^{2}+\frac{\alpha_{k}}{2q_{k}^{2}}\right),
    \label{eq:cagedoscillatornd}
\end{equation}
where $l_{k}$ are integers and $\alpha_{k}$ and $\omega$ real numbers,
is trivial.
Using our method we obtain that this system can be superintegrable
whereas using the Painlev\'e analysis we obtain that the behavior
near a movable starting point $t_{0}$ is algebraic.
In both cases the required expansion is in the form \eqref{eq:cagedoscillatorvi}
and \eqref{eq:cagedexp} respectively, with $k$ varying in
$\Set{1,2,\dots,n}$.
Indeed, in \cite{Rodriguez2008} it was showed that the $n$-dimensional 
caged anisotropic oscillator \eqref{eq:cagedoscillatornd} is maximally 
superintegrable.

\noindent In the next Subsection we will discuss another known example where our method suggest maximal superintegrability.

\subsection{The TTW system}
\label{sec:ttw}

In this Subsection we apply the method outlined above
to the Tremblay-Turbiner-Winternitz (TTW) system \cite{TTW2009}, 
namely:
\begin{equation}
    H_{\text{TTW}} = \frac{1}{2} \biggl(p_{r}^{2}+ \frac{p_{\varphi}^{2}}{r^{2}}\biggl)
+ \frac{1}{2}\omega^{2}r^{2}
+ \frac{k^{2}}{2r^{2}}\left( \frac{\alpha}{\cos^{2}(k\varphi)}
+ \frac{\beta}{\sin^{2}(k\varphi)}\right),
\label{eq:ttw}
\end{equation}
where  $r$ and $\varphi$ are polar coordinates and 
$p_{r}$ and $p_{\varphi}$ the associated generalized momenta.
The TTW system \eqref{eq:ttw} was introduced in \cite{TTW2009} where
it was shown to be integrable.
Moreover the authors analyzing the structure of the
solutions suggested that the system should have been superintegrable
(and hence maximally superintegrable since we are in the plane)
for every rational $k$.

The proof of superintegrability was accomplished successively, first for odd $k$ \cite{Quesne2010}
and then in the general case for rational $k$ 
\cite{KalninsKress2010,KalninsMillerPogosyan2011}.
In this paper we will not use the original formulation of the
TTW system. Instead, we will use the following one
presented in \cite{RodriguezTempestaTTW}:
\begin{equation}
    \mathcal{H}_{\text{TTW}}=
    k^2 (q_{1}^2+q_{2}^2)^{\frac{k-1}{k}} 
    \left[\half \left(p_{1}^2+p_{2}^2\right)
    +\frac{\omega^2}{2k^2} \left(q_{1}^2+q_{2}^2\right)^{\frac{2}{k}-1}
    +\frac{\alpha}{2q_{1}^2}+\frac{\beta}{2q_{2}^2}\right].
    \label{eq:ttwconf}
\end{equation}
The Hamiltonian \eqref{eq:ttwconf} is obtained from the Hamiltonian
\eqref{eq:ttw} from two successive canonical transformations.
The first one is the following:
\begin{subequations}
    \begin{align}
        r \cos \varphi = \frac{z+\bar{z}}{2}, 
        &\quad
        \cos\varphi \,p_{r} - \frac{\sin\varphi}{r} \,p_{\varphi}=p_{z}+p_{\bar{z}},
        \label{eq:ttwcan1a}
        \\
        r \sin\varphi = \frac{z-\bar{z}}{2\imath}, 
        &\quad
        \sin\varphi\, p_{r} +\frac{\cos\varphi}{r} \,p_{\varphi}=\text{i}(p_{z}-p_{\bar{z}}),
        \label{eq:ttwcan1b}
    \end{align}
    \label{eq:ttwcan1}
\end{subequations}
where $z$ and $\bar{z}$ are complex conjugate variables.
The application of the transformation \eqref{eq:ttwcan1} to
the Hamiltonian \eqref{eq:ttw} yields the new Hamiltonian:
\begin{equation}
    H_{\text{TTW}}^{\C} = 2 p_{z}p_{\bar{z}}
    +\frac{\omega^{2}}{2} z\bar{z}
    + \frac{2\alpha k^{2} z^{k-1}\bar{z}^{k-1}}{\left( z^{k}+\bar{z}^{k} \right)^{2}}
    - \frac{2\beta k^{2} z^{k-1}\bar{z}^{k-1}}{\left( z^{k}-\bar{z}^{k} \right)^{2}}.
    \label{eq:ttwcompl}
\end{equation}
At this point, using the canonical transformation:
\begin{subequations}
    \begin{align}
        q_{1} = \half\left( z^{k} +\bar{z}^{k} \right),
        &\quad
        p_{1} = \frac{1}{k} \left( z^{1-k}p_{z} + \bar{z}^{1-k} p_{\bar{z}} \right),
        \label{eq:ttwcan2a}
        \\
        q_{2} = \frac{1}{2\imath}\left( z^{k} -\bar{z}^{k} \right),
        &\quad
        p_{2} = \frac{\imath}{k} \left(z^{1-k} p_{z}-\bar{z}^{1-k}p_{\bar{z}} \right),
        \label{eq:ttwcan2b}
    \end{align}
    \label{eq:ttwcan2}
\end{subequations}
in the Hamiltonian \eqref{eq:ttwcompl} we obtain 
the Hamiltonian \eqref{eq:ttwconf}.
The form \eqref{eq:ttwconf} is preferable for our analysis, since
the coordinates $\left( q_{1},q_{2} \right)$ are unbounded
and take values in $\R^{2}$. 

\noindent As noted in \cite{RodriguezTempestaTTW}
the form \eqref{eq:ttwconf}, apart from the ``conformal''
factor $\left( q_{1}^{2}+q_{2}^{2} \right)^{\frac{k-1}{k}}$,
is a caged isotropic nonlinear oscillator \cite{Evans1990,Evans2008}.
Now, we will show an argument which can suggest the superintegrability
for rational $k$ of the TTW system in the form \eqref{eq:ttwconf}.
The Lagrangian corresponding to \eqref{eq:ttwconf} is given by:
\begin{equation}
    \mathcal{L}_{\text{TTW}} =
    \frac{1}{2k^{2}}  
    \left( q_{1}^{2}+q_{2}^{2} \right) ^{-{\frac {k-1}{k}}}
    \left( \dot{q}_{1}^{2}+\dot{q}_{2}^{2} \right)
    -\frac{\omega^{2}}{2}  \left( q_{1}^{2}+q_{2}^{2} \right) ^{\frac{1}{k}} 
    -\frac{{k}^{2}}{2} \left( q_{1}^{2}+q_{2}^{2} \right)^{{\frac {k-1}{k}}}
    \left( \frac{\alpha}{q_{1}^{2}}+\frac{\beta}{q_{2}^{2}} \right) .
\label{eq:lagrttw}
\end{equation}
Its Euler-Lagrange equations are:
\begin{subequations}
\begin{align}
    \begin{aligned}
        \phantom{+k}\left(  q_{1}^{2}+q_{2}^{2} \right) ^{{\frac {1}{k}}-1} 
        \left( \frac{1}{k^{2}}\ddot{q}_{1}+ \frac{{\omega}^{2}}{k} q_{1}\right) 
        &-\frac{1}{k^{3}} \left( k-1 \right) \left(  q_{1}^{2}+ q_{2}^{2} \right) ^{-{\frac {2k-1}{k}}}   
        \left[ q_{1}  \left( \dot{q}_{1}^{2} - \dot{q}_{2}^{2} \right)
        +2  q_{2} \dot{q}_{1}\dot{q}_{2}  \right]
        \\
        &-k \left(q_{1}^{2}+ q_{2}^{2} \right) ^{-{\frac {1}{k}}} 
        \left[ \frac{\alpha}{q_{1}^{3}}
        \left(q_{1} ^{2}+k q_{2}^{2} \right) 
        -\frac{\beta}{q_{2}^{2}}
        \left( k-1 \right) q_{1} 
        \right]=0,
    \end{aligned}
    \label{eq:elttwa}
    \\
    \begin{aligned}
        \phantom{+k}\left(  q_{1}^{2}+q_{2}^{2} \right) ^{{\frac{1}{k}-1}} 
        \left( \frac{1}{k^{2}}\ddot{q}_{2}+\frac{{\omega}^{2}}{k}q_{2} \right)  
       & -\frac{1}{k^{3}} \left( k-1 \right) \left(  q_{1}^{2}+ q_{2}^{2} \right) ^{-{\frac {2k-1}{k}}}   
        \left[ q_{2}  \left( \dot{q}_{2}^{2} - \dot{q}_{1}^{2} \right)
        +2  q_{1} \dot{q}_{1}\dot{q}_{2}  \right]
        \\
        &+k \left(q_{1}^{2}+ q_{2}^{2} \right) ^{-{\frac {1}{k}}} 
        \left[ 
        \frac{\alpha}{q_{1}^{2}}
        \left( k-1 \right) q_{2} 
        -\frac{\beta}{q_{2}^{3}}
        \left(k q_{1} ^{2}+q_{2}^{2} \right) 
        \right]=0.
    \end{aligned}
    \label{eq:elttwb}
\end{align}
\label{eq:elttw}
\end{subequations}
\noindent Restricting to the case $\alpha,\beta>0$ we can define 
$\alpha \doteq A^{4}$ and $\beta \doteq B^{4}$ with $A,B>0$.
In this case the \emph{real} equilibrium positions of the TTW 
system \eqref{eq:lagrttw} are given by:
\begin{equation}
q_{1}^{(\pm)} =\pm A \left( \frac{k}{\omega} \right)^{\frac{k}{2}}
\left( A^{2}+B^{2} \right)^{\frac{k-1}{2}},
\quad
q_{2}^{(\pm)} =\pm B \left( \frac{k}{\omega} \right)^{\frac{k}{2}}
\left( A^{2}+B^{2} \right)^{\frac{k-1}{2}}.
\label{eq:ttwequil}
\end{equation}
Due to the fact that the system \eqref{eq:lagrttw} is symmetric
under the discrete transformations $q_{1}\to-q_{1}$ and
$q_{2}\to-q_{2}$ we can consider only the equilibrium positions
labeled by $(+,+)$ in \eqref{eq:ttwequil}.
From equation \eqref{eq:pertv} we introduce:
\begin{equation}
q_{1} = q_{1}^{(+)} + \varepsilon Q_{1},
\quad
q_{2} = q_{2}^{(+)} + \varepsilon Q_{2}.
\label{eq:vttw}
\end{equation}
Inserting equation \eqref{eq:vttw} into the 
Euler-Lagrange equations \eqref{eq:elttw} 
and expanding in Taylor series with respect to $\varepsilon$,
we obtain as coefficients of $\varepsilon$
the following linearized equations:
\begin{subequations}
\begin{align}
    \ddot{Q}_{1} +{\frac {4{\omega}^{2}
    \left( {A}^{2}+{k}^{2}{B}^{2} \right)}{{A}^{2}+{B}^{2}}} Q_{1}
    +{\frac {4{\omega}^{2}AB \left( 1-{k}^{2} \right) }{{A}^{2}+{B}^{2}}} Q_{2}=0 \, ,
    \\
    \ddot{Q}_{2} + 
    {\frac {4{\omega}^{2}AB \left( 1-{k}^{2} \right)}{{A}^{2}+{B}^{2}}} Q_{1}
    +{\frac {4{\omega}^{2} \left( {k}^{2}{A}^{2}+{B}^{2} \right)}{{A}^{2}+{B}^{2}}}Q_{2}=0 \, .
    \label{eq:ellinttwb}
\end{align}
\label{eq:ellinttw}
\end{subequations}
We introduce the normal coordinates $N_{1}$ and $N_{2}$
through the linear transformation:
\begin{equation}
\begin{bmatrix}
Q_1     \\
Q_2      
\end{bmatrix}=\begin{bmatrix}
A       &  B\\
B & -A \\
\end{bmatrix}\begin{bmatrix}
N_1     \\
N_2      
\end{bmatrix}
\label{eq:lintransfttw}
\end{equation}
obtaining from \eqref{eq:ellinttw}
the following linearized system:
\begin{equation}
    \ddot{N}_{1}+4\omega^{2} N_{1}=0,
    \quad
    \ddot{N}_{2}+4k^{2}\omega^{2} N_{2}=0.
\label{eq:ellintransfttw}
\end{equation}
In this case the fundamental frequencies
are $\boldsymbol{\omega} = \left( 2\omega, 2 k \omega \right)$, which
means that we can expect (maximal) superintegrability only if
$k$ is rational, just as suggested in \cite{TTW2009}.
Now, we use a multiple-scale expansion to show that 
at higher orders the periodicity is preserved.
Thus, this time we perform a multiple-scale
expansion in terms of the normal coordinates
$N_{1}$ and $N_{2}$, with the 
three trivial time scales $t_{i}$ \eqref{eq:trivialts} with $i=0,1,2$:
\begin{equation}
    N_{k} = N_{k}^{(0)}\left( t_{0},t_{1},t_{2} \right)
    +\varepsilon N_{k}^{(1)}\left( t_{0},t_{1},t_{2} \right)
    +\varepsilon^{2} N_{k}^{(2)}\left( t_{0},t_{1},t_{2} \right)
    + O\left( \varepsilon^{3} \right),
\quad k=1,2.
\label{eq:vtpertttw}
\end{equation}
Then, we can insert this expansion in (\ref{eq:vttw}-\ref{eq:lintransfttw})
and into \eqref{eq:elttw}. Expanding in Taylor series with respect to
$\varepsilon$ and taking coefficients up to the second order
yields a sequence of systems which can readily be solved.
Surprisingly enough, if we do not know the properties of the
TTW system \eqref{eq:ttwconf}, we find that there are no corrective terms 
to the fundamental frequencies and the
asymptotic solution, valid up to $\varepsilon^{2} t=\BigO{1}$,
is just given by:
\begin{subequations}
\begin{align}
    N_{1} &= R_{1}\cos\left( 2\omega t + \varphi_{1} \right)
    +\BigO{\varepsilon},
    \label{eq:ttwv2a}
    \\
    N_{2} &= R_{2} \cos\left( 2k\omega t + \varphi_{2} \right)
    +\BigO{\varepsilon},
    \label{eq:ttwv2b}
\end{align}
\label{eq:ttw2}
\end{subequations}
where $R_{k}$ and $\varphi_{k}$ with $k=1,2$ are integration constants. 
This just shows that the ratio is preserved and that the
TTW system \eqref{eq:ttwconf} can be superintegrable.
Moreover, this also shows that the multiple-scale expansion
collapses into a standard perturbative expansion of the form:
\begin{equation}
    N_{k}(t) = \sum_{i=0}^{\infty} \varepsilon^{i}N_{k}^{(i)}(t),\quad
    k=1,2.
\label{eq:ttwstadnexp}
\end{equation}
The fact that there are no corrections to the fundamental
frequencies reflects the property of isochronicity of the
TTW system, since the frequency of oscillation of the
system is independent from the initial values \cite{TTW2010,Calogero2008book}. Without giving a complete account of the form of the orbits of the TTW system \cite{TTW2010,RodriguezTempestaTTW}, 
we show some numerical examples in Figure \ref{fig:ttw}.

\begin{figure}[]
\centering
\subfloat[][$k=1$]{%
    \includegraphics[width=0.5\textwidth]{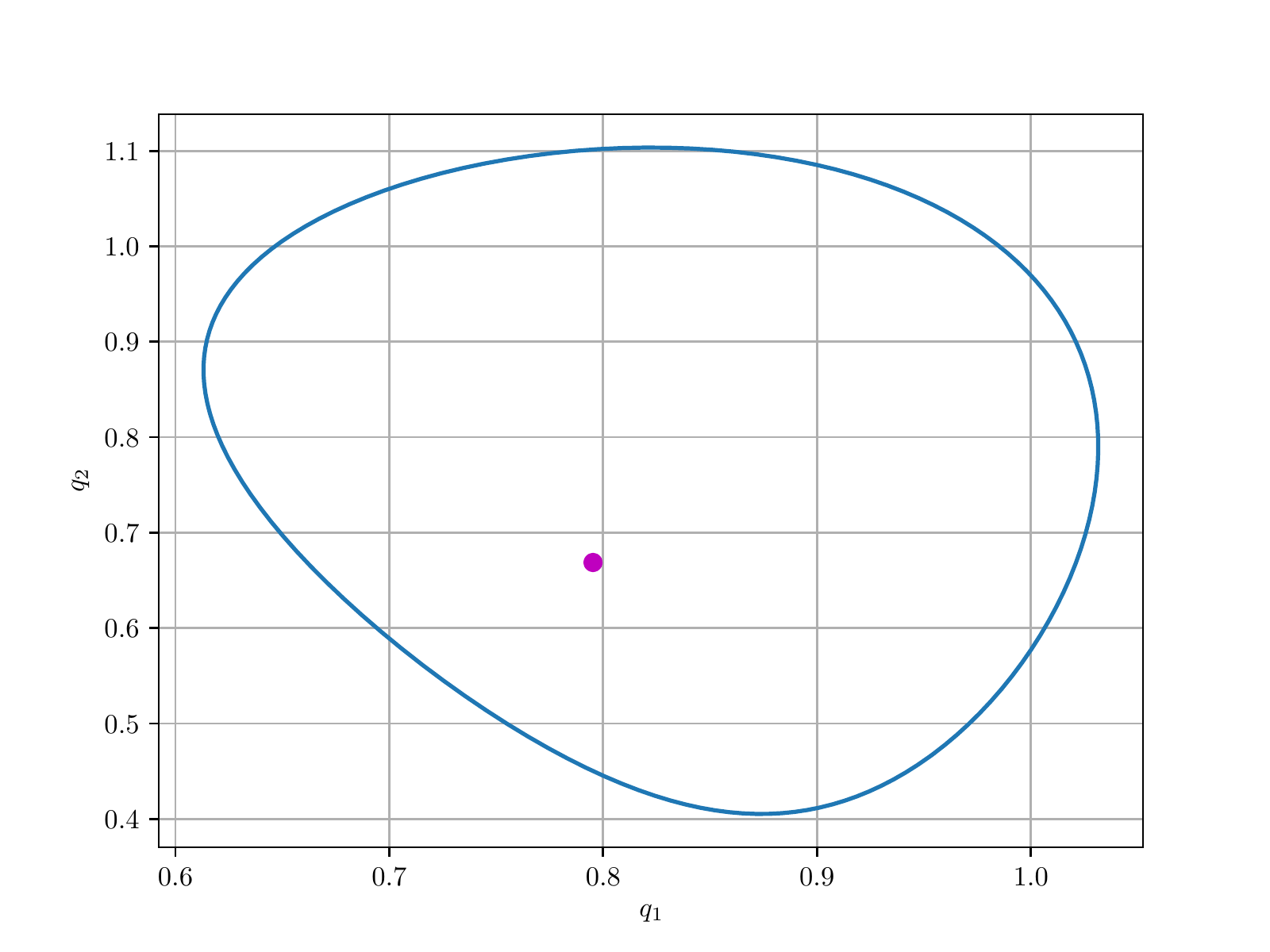}
    }
\subfloat[][$k=2$]{
    \includegraphics[width=0.5\textwidth]{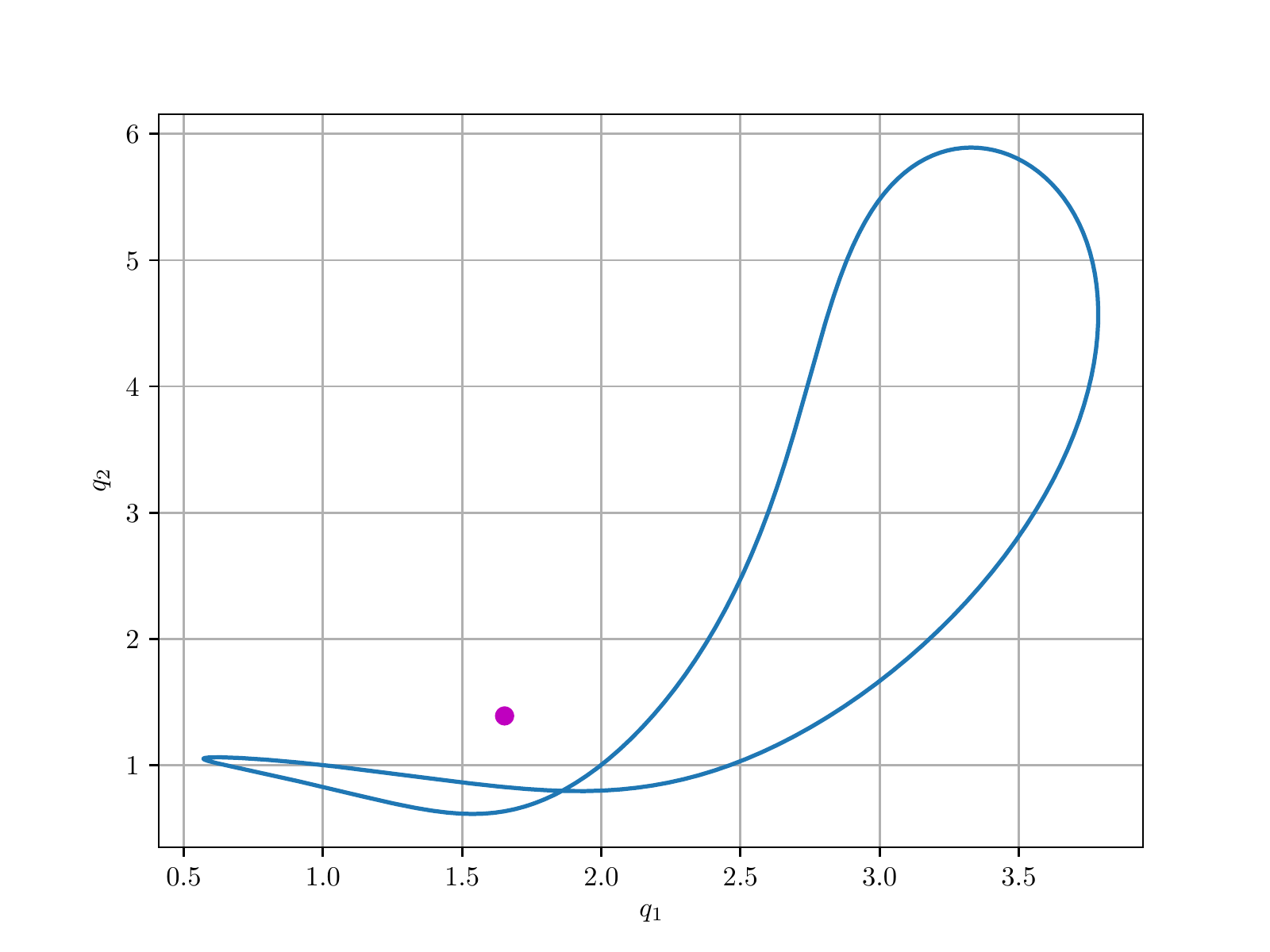}
}
\\
\subfloat[][$k=4/3$]{%
    \includegraphics[width=0.5\textwidth]{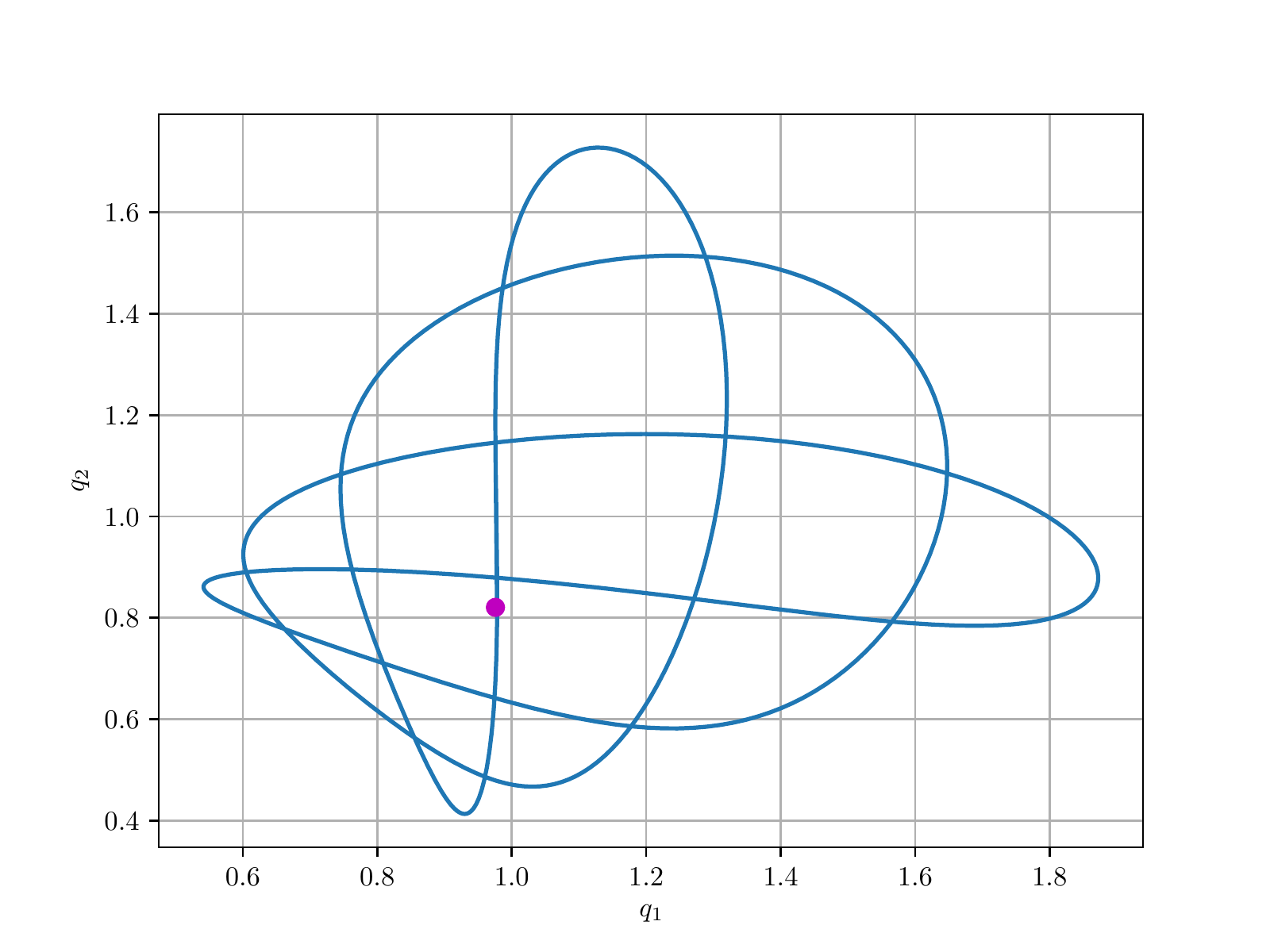}
    }
\subfloat[][$k=4/5$]{
    \includegraphics[width=0.5\textwidth]{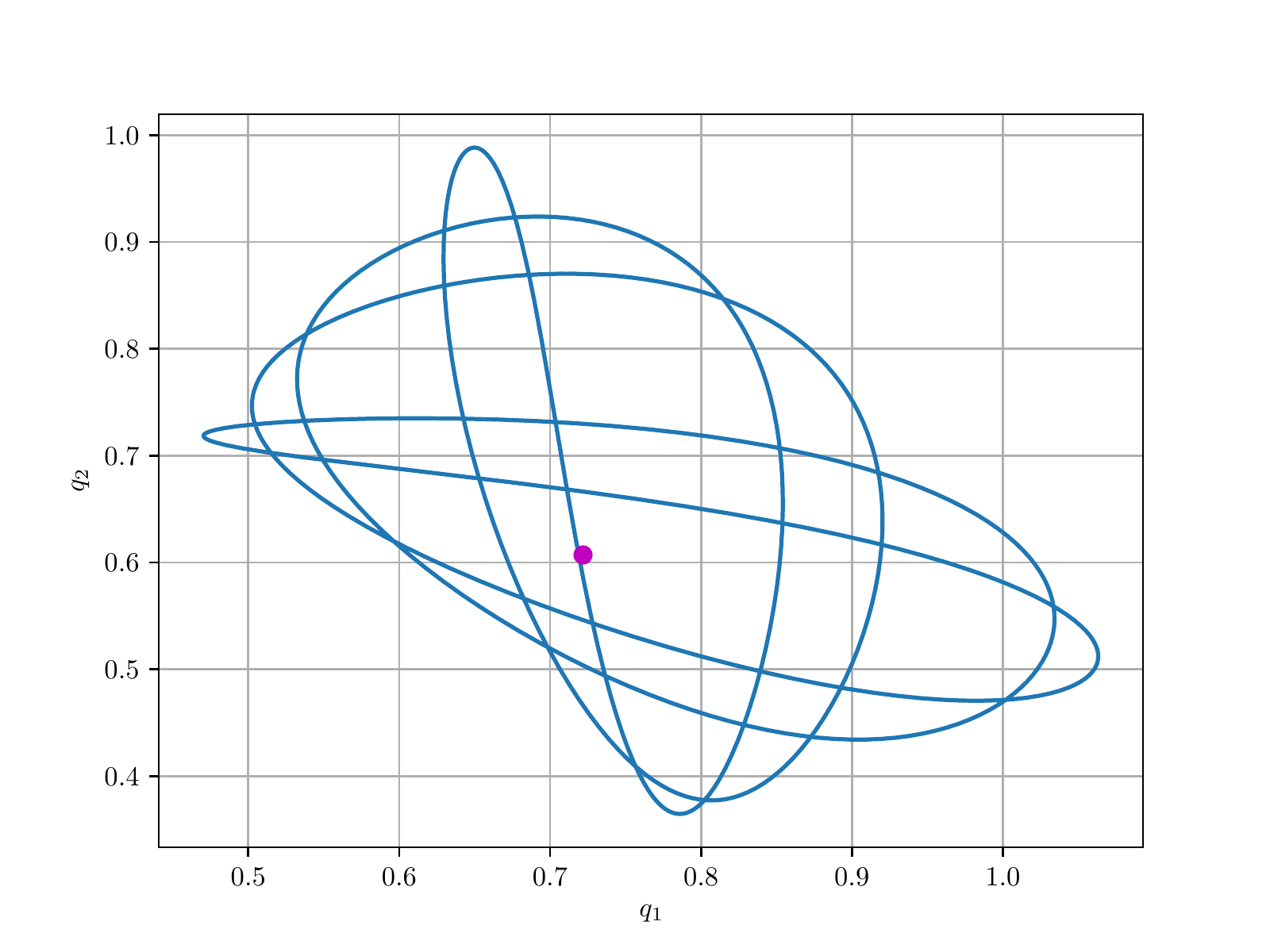}
}
\\
\subfloat[][$k=5/3$]{
    \includegraphics[width=0.5\textwidth]{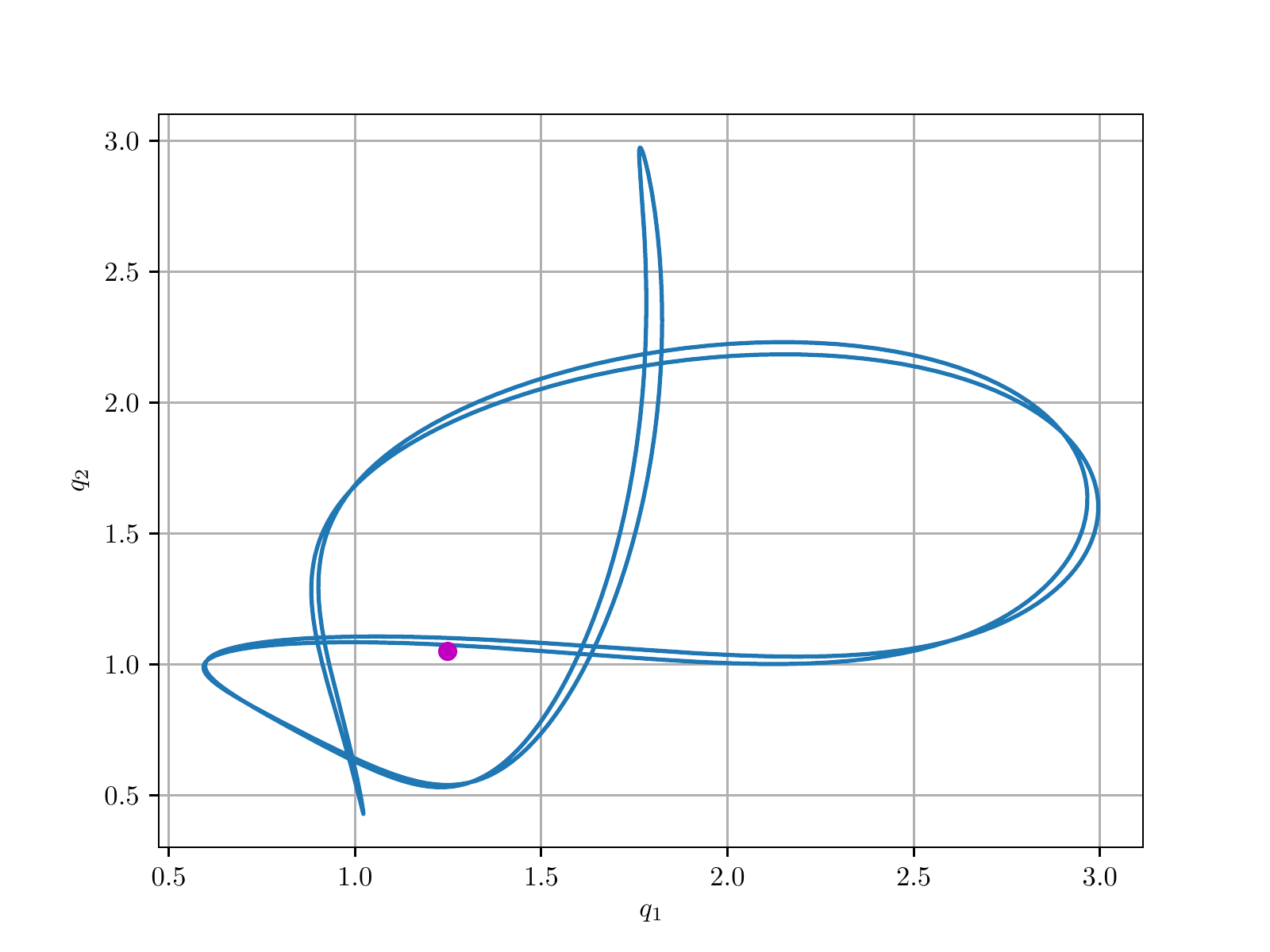}
}
\subfloat[][$k=5/7$]{
    \includegraphics[width=0.5\textwidth]{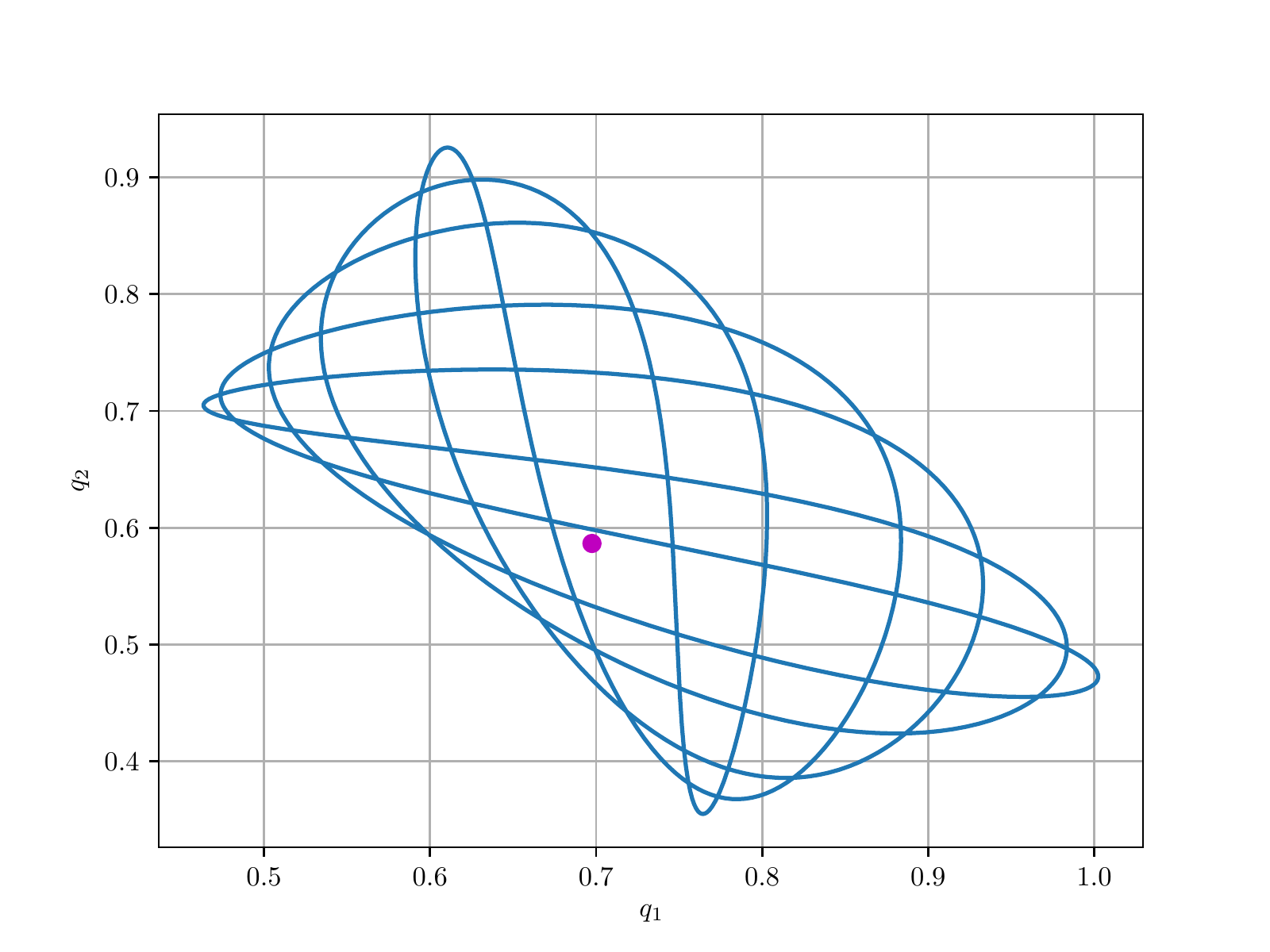}
}

\caption{Trajectories of the TTW Hamiltonian
    \eqref{eq:ttwconf} for different values of $k=k_{\text{n}}/k_{\text{d}}$
    with integer $k_{\text{n}}$ and $k_{\text{d}}$ with
    fixed parameters $\omega=1/2$, $\alpha=1/10$ and $\beta=1/20$.
    We used the \texttt{odeint} integrator from \texttt{scipy} \cite{scipy}
    with a regular mesh of $N=2^{25}$ points in the interval 
    $[0,2\pi k_{\text{n}}k_{\text{d}} /\omega]$.
    The magenta dots are the stable equilibrium points labeled by
    $\left( +,+ \right)$ in \eqref{eq:ttwequil}.}
\label{fig:ttw}
\end{figure}

\subsection{The Drach system}
\label{sec:drach}

In this Subsection we present a new result about what we are going to call the \emph{Drach system}.
In 1935 Drach \cite{Drach1935} carried out the first systematic
search for integrable systems with a third-order integral of
motion in two dimensions. He conducted his search in a
two-dimensional complex space $\mathbb{E}_{2}\left( \C \right)$
and found ten potentials. These potentials have been subject
of further investigations in more recent times 
\cite{Tsiganov2000,Tsiganov2008,TondoTempesta2016}.
In particular one of these potentials can be rewritten 
in real form as:
\begin{equation}
H_{\text{Drach}}=\frac{1}{2}\left(p_1^2+ p_2^2\right)
+\frac{a q_2+b+c (4 q_1^2+3 q_2^2)}{q_1^{2/3}},
\label{eq:Hdrach}
\end{equation}
where \emph{a priori} $a$, $b$ and $c$ are real parameters.
In \cite{PostWinternitz2011} a particular case
of \eqref{eq:Hdrach} with $c=0$ was considered and showed to be
(maximally) superintegrable with a third-order and 
a fourth-order additional integral of motion. 

We will refer to this particular case as the  \emph{Drach-Post-Winternitz} (DPW) system.
The DPW system is important since
its quantum version was the first \emph{nonseparable}
(maximally) superintegrable system ever found.
In \cite{NucciPost2012}, using the so-called \emph{reduction method} \cite{Nucci1996b,Nucci2000,Nucci2001}, it was shown to be reducible to the third-order trivial equation $X'''=0$.

At the end of their paper the authors in \cite{PostWinternitz2011}
state that:
\begin{quote}
``For $c\neq 0$ (31) [here \eqref{eq:Hdrach}] does not allow 
a fourth-order integral, though it still might be superintegrable.''
\end{quote}
leaving open the question whether or not the ``full'' Drach system
\eqref{eq:Hdrach} is (maximally) superintegrable.
Now, we show using the argument proposed in Section \ref{sec:method}
that the Drach system is in general \emph{not} (maximally) superintegrable.
Indeed, the Lagrangian corresponding to \eqref{eq:Hdrach} is:
\begin{equation}
L_{\text{Drach}}=\frac{1}{2}\left(\dot q_1^2+\dot q_2^2\right)-\frac{a q_2+b+c (4 q_1^2+3 q_2^2)}{q_1^{2/3}}.
\label{eq:drach}
\end{equation}
The corresponding Euler-Lagrange equations are:
\begin{subequations}
\begin{align}
    \ddot q_{1} + 8 c q_1^{1/3}
    -\frac{2}{3} \frac{a q_2+b+c (4 q_1^2+3 q_2^2)}{q_1^{5/3}}&=0,
    \label{eq:dracheqa}
    \\
    \ddot q_{2} + \frac{a+6 c q_2}{q_1^{2/3}} &= 0.
    \label{eq:dracheqb}
\end{align}
\label{eq:dracheq}
\end{subequations}
\noindent We have two equilibrium positions:
\begin{equation}
q_{1}^{(\pm)}= \pm\frac{\eta}{24c}, \quad
q_{2}^{(0)} = -\frac{a}{6c},
\label{eq:equil}
\end{equation}
where we have introduced the shorthand notation:
\begin{equation}
72 b c -6 a^2 =\eta^2.
\label{eq:etadef}
\end{equation}
The presence of $b$ here is necessary in order to have
\emph{real} solutions to the equations of the equilibrium conditions.
Therefore, from now on we will assume $b > a^{2}/(12c)$
and $c\neq 0$, otherwise \eqref{eq:equil} loses its meaning.
We observe that the $c=0$ is just the DPW case
\cite{PostWinternitz2011}, which has no equilibrium solutions.  
From equation \eqref{eq:pertv} we obtain:
\begin{equation}
    q_{1} = q_{1}^{(\pm)} + \varepsilon Q_{1}, \quad
    q_{2} = q_{2}^{(0)} + \varepsilon Q_{2}
\label{eq:pert}
\end{equation}
with $q_{1}^{(\pm)}$ and $q_{2}^{(0)}$ given by \eqref{eq:equil}.
The equilibrium position with $q_{1}^{(-)}$ is not stable for
every value of the parameters, whereas the one with $q^{(+)}$, for $c>0$,
yields to the following linearized equations:
\begin{subequations}
\begin{align}
    \ddot Q_{1} + \frac{128}{3}\left(\frac{3}{\eta}\right)^{2/3}
        c^{5/3} Q_{1} &=0,
    \label{eq:drachlina}
    \\
    \ddot Q_{2} +24\left(\frac{3}{\eta}\right)^{2/3}
        c^{5/3}  Q_{2} &=0.
    \label{eq:drachlinb}
\end{align}
\label{eq:drachlin}
\end{subequations}
The system is already in normal form and the
fundamental frequencies are given by:
\begin{equation}
\omega_{1}^{(0)} = \frac{8\sqrt{6}c^{5/6} 3^{1/3}}{3 \eta^{1/3}},
\quad
\omega_{2}^{(0)} = \frac{2 \sqrt{6} c^{5/6} 3^{1/3}}{\eta^{1/3}}.
\label{eq:fundfreq}
\end{equation}
The ratio of the two fundamental frequencies is rational:
\begin{equation}
    \Delta^{(0)}= \frac{\omega_{1}^{(0)}}{\omega_{2}^{(0)}} = \frac{4}{3}.
\label{eq:fundrat}
\end{equation}
At this stage we have that the Drach system \eqref{eq:Hdrach}
might be (maximally) superintegrable. 
To check what happens at higher
order we use a multiple-scale expansion with three trivial
time scales \eqref{eq:trivialts} of the form \eqref{eq:HHvpert}.
We substitute \eqref{eq:HHvpert} into \eqref{eq:pert} and then
into the Euler-Lagrange equations \eqref{eq:dracheq}.
Expanding in Taylor series with respect to $\varepsilon$, and then
carrying out the computations, we get the following solution:
\begin{subequations}
\begin{align}
    Q_{1} = R_{1} 
    \cos\left[ 
        \left( \omega_{1}^{(0)}+\omega_{1}^{(2)}\varepsilon^{2} \right)t
    +\varphi_{1}\right]
    +\BigO{\varepsilon},
    \label{eq:v1mult}
    \\
    Q_{2} = R_{2} 
    \cos\left[ 
        \left( \omega_{2}^{(0)}+\omega_{2}^{(2)}\varepsilon^{2} \right)t
    +\varphi_{2}\right]
    +\BigO{\varepsilon}.
    \label{eq:v2mult}
\end{align}
\label{eq:vimult}
\end{subequations}
Here $R_{i}$ and $\varphi_{i}$ are constants of integration,
$\omega_{i}^{(0)}$ are given by \eqref{eq:fundfreq}
and the corrections to the frequencies are given by:
\begin{subequations}
\begin{align}
    \omega_{1}^{(2)} &=
    -\frac{8}{45}\frac{\sqrt{6}c^{17/6} 3^{1/3} 
        \left(400 R_1^2+513 R_2^2\right)}{\eta^{7/3}},
    \label{eq:omeps2a}
    \\
    \omega_{2}^{(2)} &=
    -\frac{2}{5}\frac{\sqrt{6} c^{17/6} 3^{1/3}
        \left(304 R_1^2+54 R_2^2\right)}{\eta^{7/3}}.
    \label{eq:omeps2b}
\end{align}
\label{eq:omeps2}
\end{subequations}
The ratio of these additional frequencies is no longer given
by \eqref{eq:fundrat} and in fact we can see that:
\begin{equation}
    \Delta^{(2)}=\frac{\omega_{1}^{(0)}+\omega_{1}^{(2)}\varepsilon^{2}}{%
        \omega_{2}^{(0)}+\omega_{2}^{(2)}\varepsilon^{2}}
= \frac{4}{3} + \frac{4}{45}\frac{c^{2}}{\eta^{2}}
\left( 512 R_{1}^{2}-351 R_{2}^{2} \right)\varepsilon^{2}
+ \BigO{\varepsilon^{4}}.
\label{eq:freqrat}
\end{equation}
As discussed in Section \ref{sec:method} this behavior
is not consistent with a (maximally) superintegrable system.
Therefore, we are led to conclude that \emph{in general}
the Drach system is an integrable, but not a superintegrable system.
This does not exclude the possibility of the existence of other
(maximally) superintegrable subcases, especially in the
range of parameters which cannot be covered with our technique.
For example, it is worth mentioning that the DPW system,
which arises as singular limit as $c\to0$, does not possess closed
orbits. This leaves open the question whether or not there exists
other superintegrable subcases without closed orbits. To give the feeling of the form of the trajectories of the Drach system, and of the fact that they are non-periodic we show some numerical examples in Figure \ref{fig:numstud}.
\begin{figure}[h!bt]
	\centering
	\subfloat[][$t=3\omega_1$]{%
		\includegraphics[width=0.5\textwidth]{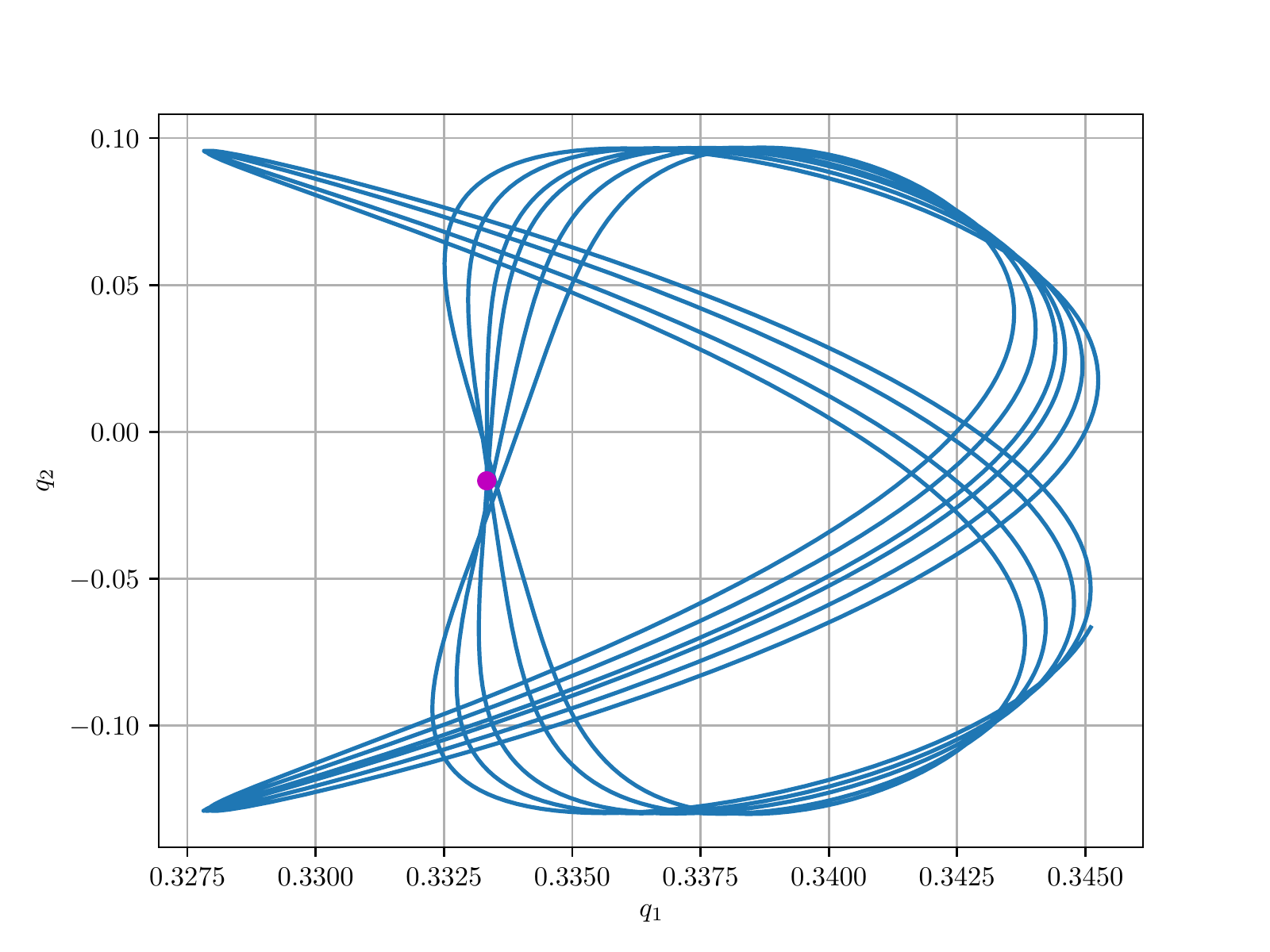}
	}
	\subfloat[][$t=6\omega_{1}$]{
		\includegraphics[width=0.5\textwidth]{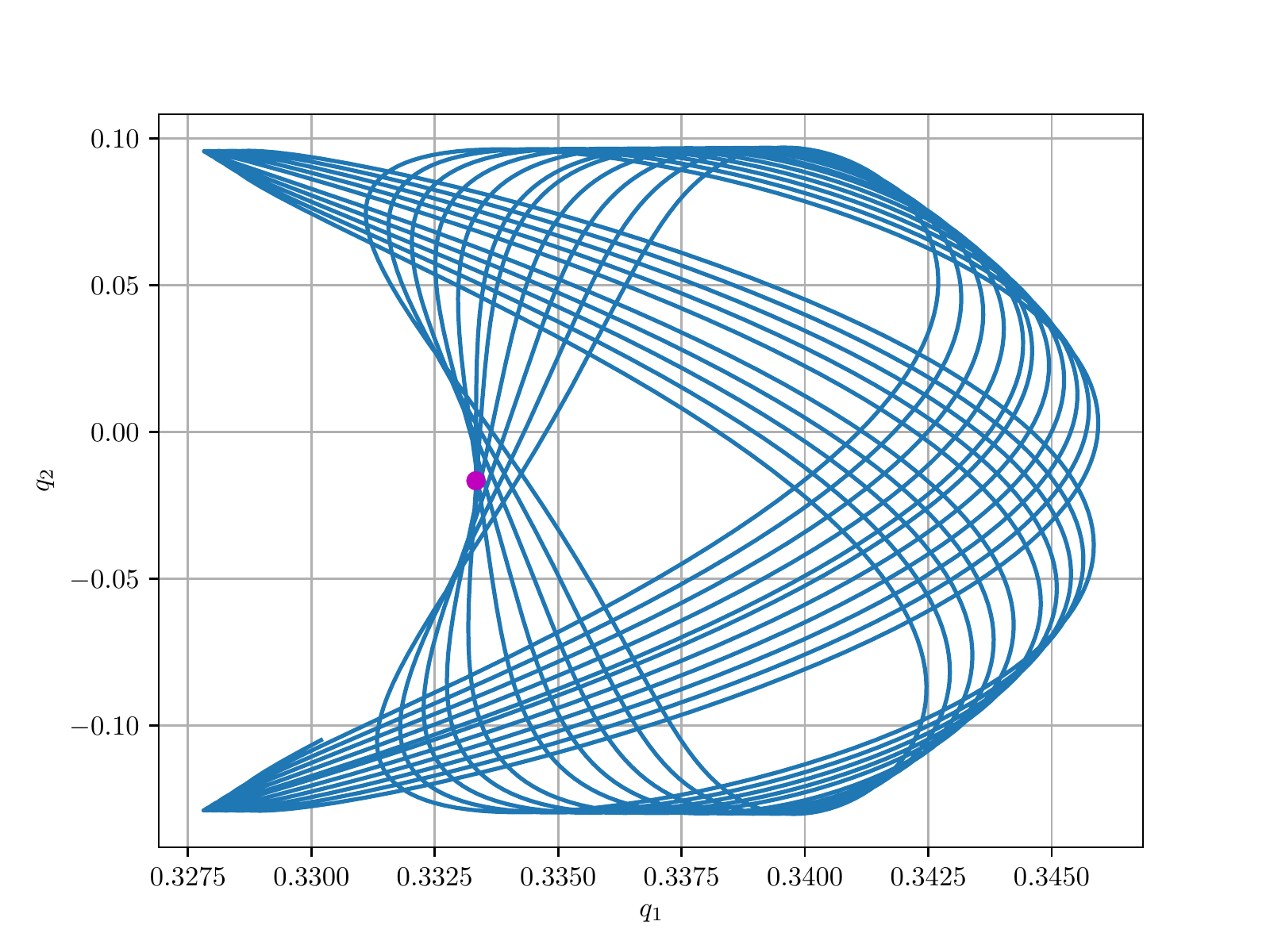}
	}
	\\
	\subfloat[][$t=9\omega_{1}$]{
		\includegraphics[width=0.5\textwidth]{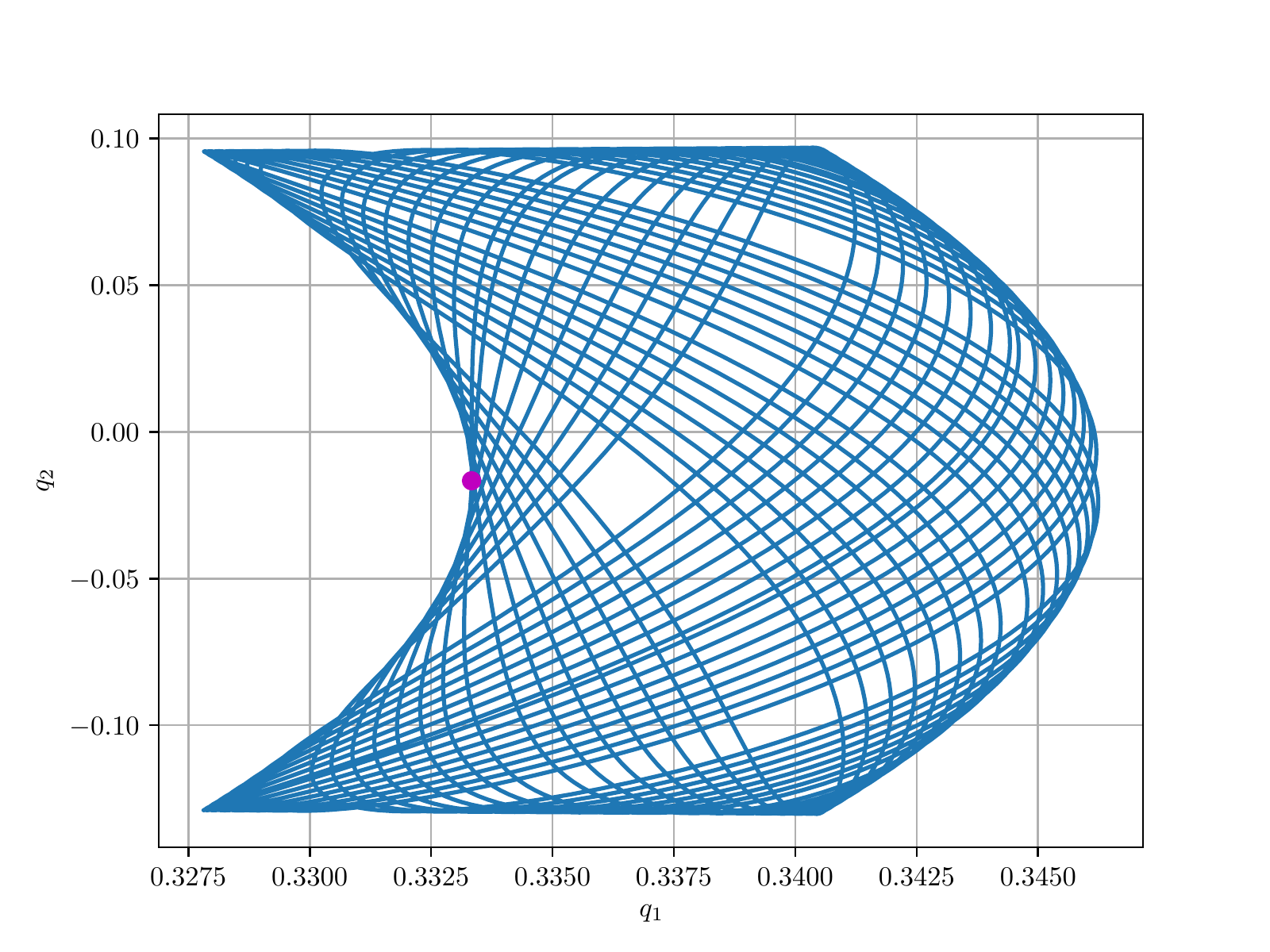}
	}
	\subfloat[][$t=12\omega_{1}$]{
		\includegraphics[width=0.5\textwidth]{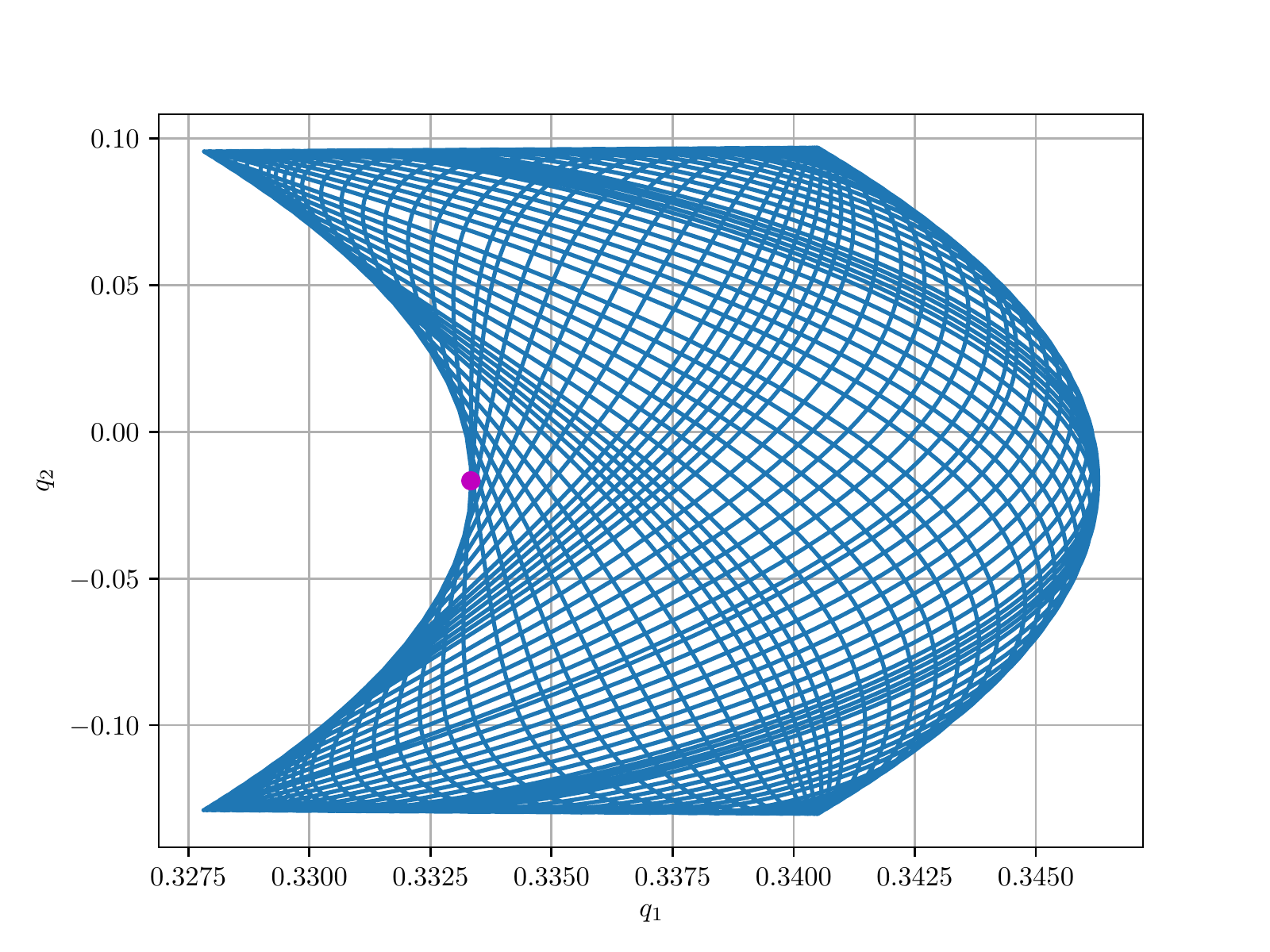}
	}
	\\
	\caption{Trajectories for the maps \eqref{eq:dracheq} with parameters
		$a=0.1$, $\eta=8$, $c=1$. We used the \texttt{odeint}
		integrator from \texttt{scipy} \cite{scipy} with a regular mesh of $N=2^{25}$ 
		points.
		The magenta dot is the stable equilibrium point labeled
		by $(+)$ in \eqref{eq:equil}.}
	\label{fig:numstud}
\end{figure}

To conclude this paper we wish to show with an explicit example how
this approach can be used to suggest maximal superintegrability of
a class of Hamiltonian systems.
We will give a proof of Bertrand's theorem 
\cite{Bertrand1873} using our method. 
This will also point out the main differences between our approach
and the one followed in \cite{Zarmi2002}.

\subsection{Central force potentials:  Bertrand's theorem}

\label{sss:bertrand}

In this Subsection we give an alternative proof of the following:

\begin{theorem}[Bertrand \cite{Bertrand1873}]
	Assume we are given the Hamiltonian in the Euclidean plane:
	\begin{equation}
	H_{k,\alpha} = \half \left( p_{r}^{2} + \frac{p_{\varphi}^{2}}{r^{2}} \right)
	+ k r^{\alpha}
	\label{eq:hrp}
	\end{equation}
	where $r\in\R^{+}$ and $\varphi\in[0,2\pi)$ are polar coordinates, $p_{r}$
	and $p_{\varphi}$ the respective momenta, and $k$ and $\alpha$ 
	are real constants.
	Then, the only two cases for which all bounded orbits are closed  
	are those corresponding to $\alpha=2$, $k>0$ and $\alpha=-1$,
	$k<0$, i.e. the isotropic harmonic oscillator and the 
	Kepler system.
	\label{thm:bertrand}
\end{theorem}

\begin{proof}
	We will give a direct proof of this theorem, i.e. without resorting
	to the equation of the orbits \cite{Whittaker,LeviCivita1922vol1,Appell1893}, by using the multiple scales approach.
	The Lagrangian associated with \eqref{eq:hrp} is:
	\begin{equation}
	L_{k,\alpha} = \half \left( \dot{r}^{2} + r^{2}\dot{\varphi}^{2} \right)
	- k r^{\alpha}.
	\label{eq:lrp}
	\end{equation}
	The Euler-Lagrange equations corresponding to \eqref{eq:lrp} are
	given by:
	\begin{subequations}
		\begin{align}
		r^{2}\ddot{\phi} + 2 r \dot{\phi}  \dot{r} &=0,
		\label{eq:phieq}
		\\
		\ddot{r} - r  \dot{\phi}^{2} + k \alpha r^{\alpha-1}   &=0.
		\label{eq:req}
		\end{align}
		\label{eq:lrpeq}
	\end{subequations}
	It is easy to check that in this case there are no equilibrium positions,
	so the method does not seem applicable.
	We recall that, in order to apply the procedure of Section \ref{sec:method}, we need
	a stable equilibrium solution to check whether the perturbative
	expansion around the stable point gives rise to closed orbits or not.
	Due to the radial symmetry of the Hamiltonian \eqref{eq:hrp} it is
	natural to make a different \emph{ansatz}: we search for the existence
	of \emph{circular orbits}, i.e. a two-parameter family of solutions of the
	form:
	\begin{equation}
	r = r_{0}, \quad \varphi = \omega_{0} t +\varphi_{0},
	\label{eq:circorb}
	\end{equation}
	where $\omega_{0}=\omega_{0}\left( r_{0} \right)$.
	The motion defined by \eqref{eq:circorb} is called circular
	since it represents the equation of a circle in polar coordinates 
	\cite{LeviCivita1922vol1}.
	Inserting the \emph{ansatz} into the equations of motion
	\eqref{eq:lrpeq} we obtain:
	\begin{equation}
	\omega_{0} = \sqrt{k\alpha} r_0^{\alpha/2-1},
	\label{eq:om0}
	\end{equation}
	which implies that there exist circular orbits for the system \eqref{eq:lrpeq}
	if $k\alpha >0$, i.e. $k>0$ and $\alpha>0$ or $k<0$ and $\alpha<0$.
	We underline that circular orbits are \emph{exact solutions}
	and their property is to maximize the \emph{angular momentum}:
	\begin{equation}
	\ell = r^{2} \dot{\varphi},
	\label{eq:angmom}
	\end{equation}
	being both $r$ and $\dot{\varphi}$ constants.
	Since every other kind of orbit will have smaller angular momentum, 
	when circular orbits do not exist, i.e. when $k\alpha<0$, 
	there are no bounded trajectories.
	We note that due to the radial symmetry we can always suppose
	$\varphi_{0}\equiv0$ in \eqref{eq:circorb}.
	
	\noindent  Therefore, we can restrict the analysis to the range of parameters $k\alpha>0$.
	We turn to investigate the stability of circular orbits in this range. 
	We introduce the expansion of \emph{nearly-circular orbits}:
	\begin{equation}
	r = r_{0} + \varepsilon \rho, 
	\quad 
	\varphi = \omega_{0} \left(t + \varepsilon \theta\right),
	\label{eq:ncircorb}
	\end{equation}
	where $\rho=\rho\left( t \right)$, $\theta=\theta\left( t \right)$
	and $\varepsilon\to0^{+}$.
	Inserting the expansion \eqref{eq:ncircorb} into \eqref{eq:lrpeq},
	at the first order in $\varepsilon$ we obtain:
	\begin{subequations}
		\begin{align}
		r_{0}\ddot{\theta} +2 \dot{\rho} &=0,
		\label{eq:phieqe}
		\\
		\ddot{\rho}
		-2 r_{0}^{\alpha-1} k \alpha \dot{\theta} 
		+\alpha k \left( \alpha-2 \right) r_{0}^{\alpha-2}\rho &=0.
		\label{eq:reqe}
		\end{align}
		\label{eq:lrpeqe}
	\end{subequations}
	The system \eqref{eq:lrpeqe} has the following solution:
	\begin{subequations}
		\begin{align}
		\rho &= \rho_0 \cos(r_0^{\alpha/2-1} \sqrt{k\alpha\left(\alpha+2\right)}t+\beta_{0})
		+\frac{2 r_0 \Omega_{0}}{\alpha+2},
		\label{eq:reqesol}
		\\
		\theta &=
		\frac{\alpha-2}{\alpha+2}\Omega_{0} t
		-\frac{2\rho_{0}}{r_{0}^{\alpha/2}\sqrt{k\alpha\left(\alpha+2\right)}}
		\sin(r_{0}^{\alpha/2-1} \sqrt{k\alpha\left(\alpha+2\right)} t+\beta_{0}) 
		+\theta_0,
		\label{eq:phieqesol}
		\end{align}
		\label{eq:lrpeqesol}
	\end{subequations}
	where $\rho_{0}$, $\beta_{0}$, $\Omega_{0}$ and $\theta_{0}$ are constants
	of integration.
	The solution \eqref{eq:lrpeqesol} implies that the circular orbits 
	are stable if $k\alpha\left( \alpha+2 \right)>0$ hence, being $k\alpha>0$,
	when $\alpha>-2$.
	When $\alpha<-2$ circular orbits are unstable and we have unbounded motion. 
	Moreover, recalling formula \eqref{eq:ncircorb}, 
	we obtain that at the zero-th order the frequencies are given by $\omega_{0}$ 
	from \eqref{eq:om0} and by:
	\begin{equation}
	\bar{\omega}_{0} = r_{0}^{\alpha/2-1} \sqrt{k\alpha\left(\alpha+2\right)}.
	\label{eq:varpi0}
	\end{equation}
	Now, due to the fact that $\varphi$ is an angle variable, we have that
	at the zero-th order the condition for the orbit to be closed is
	given by $\omega_{0}/\bar{\omega}_{0}\in\Q$.
	This implies:
	\begin{equation}
	\Delta^{(0)}= \frac{1}{\sqrt{\alpha+2}} = \frac{n}{m},
	\label{eq:alphacond}
	\end{equation}
	where $n$, $m$ are co-prime integers.
	
	To check if this condition is preserved at higher order we perform 
	a multiple-scale expansion with three trivial time scales given by
	equation \eqref{eq:trivialts}:
	\begin{subequations}
		\begin{align}
		\rho &= \rho_{0}\left( t_{0},t_{1},t_{2} \right)
		+\varepsilon \rho_{1}\left( t_{0},t_{1},t_{2} \right)
		+\varepsilon^{2} \rho_{2}\left( t_{0},t_{1},t_{2} \right)
		+ \BigO{\varepsilon^{3}},
		\label{eq:rhomult}
		\\
		\theta &= \theta_{0}\left( t_{0},t_{1},t_{2} \right)
		+\varepsilon \theta_{1}\left( t_{0},t_{1},t_{2} \right)
		+\varepsilon^{2} \theta_{2}\left( t_{0},t_{1},t_{2} \right)
		+ \BigO{\varepsilon^{3}}.
		\label{eq:phimult}
		\end{align}
		\label{eq:rhophimult}
	\end{subequations}
	Inserting the expansion \eqref{eq:rhophimult} into the form
	of the nearly-circular orbits \eqref{eq:ncircorb}, and then into
	\eqref{eq:lrpeq}, we obtain the following nearly-circular orbit:
	\begin{subequations}
		\begin{align}
		r 
		& = r_{0} 
		+\varepsilon\frac{2 r_{0} \Omega_{0}}{\alpha+2}+ \varepsilon\rho_{0} 
		\cos\left[\bar{\omega}_{0} \left(1 
		+\epsilon\frac{\alpha-2}{\alpha+2}\Omega_{0} 
		+\varepsilon^{2}\bar{\mu}^{(2)}  \right) t 
		+\beta_{0}\right]
		+\BigO{\varepsilon^{2}},
		\label{eq:rhomultsol}
		\\
		\varphi &
		\begin{aligned}[t]
		&= \omega_{0} 
		\left\{1+\varepsilon \frac{\alpha-2}{\alpha+2}\Omega_{0} 
		+\frac{\varepsilon^{2}}{2} \left( \alpha-2 \right)
		\left[ \frac{1}{4}\left(\frac{\rho_{0}}{r_{0}}\right)^2 (\alpha-1) 
		+\frac{\alpha-4}{\left(\alpha+2\right)^2} \Omega_{0}^2  \right]
		\right\} t 
		\\
		&-\frac{2\varepsilon\omega_{0}\rho_{0}}{r_{0}^{\alpha/2}\sqrt{k\alpha\left(2+\alpha\right)}}
		\sin\left[ \bar{\omega}_{0}\left(1 
		+\epsilon\frac{\alpha-2}{\alpha+2}\Omega_{0} 
		+\varepsilon^{2}\bar{\mu}^{(2)}  \right) t 
		+\beta_{0}\right] 
		+\varepsilon\omega_{0}\theta_0
		+\BigO{\varepsilon^{2}},
		\end{aligned}
		\label{eq:phimultsol}
		\end{align}
		\label{eq:rphimult}
	\end{subequations}
where $\rho_{0}$, $\beta_{0}$, $\Omega_{0}$ and $\theta_{0}$ are constants
	of integration, $\omega_{0}$ is given by \eqref{eq:om0},
	$\bar{\omega}_{0}$ is given by \eqref{eq:varpi0} and:
	\begin{equation}
	\bar{\mu}^{(2)} =  \half
	(\alpha-2)\left[ \frac{1}{6} \left(\frac{\rho_{0}}{r_{0}}\right)^{2} (\alpha-2)
	+\frac{ \alpha-4}{(\alpha+2)^{2}}\Omega_{0}^2\right].
	\label{eq:varpi2}
	\end{equation}
	Defining:
	\begin{equation}
	\mu^{(2)} \doteq 
	\frac{1}{2} \left( \alpha-2 \right)
	\left[ \frac{1}{4}\left(\frac{\rho_{0}}{r_{0}}\right)^2 (\alpha-1) 
	+\frac{\alpha-4}{\left(\alpha+2\right)^2} \Omega_{0}^2  \right],
	\label{eq:om2}
	\end{equation}
	we obtain from \eqref{eq:phimultsol} 
	that, at the second order, the frequencies are given by 
	\begin{subequations}
		\begin{align}
		\omega_{2} &= \omega_{0}\left( 1+\varepsilon \frac{\alpha-2}{\alpha+2}\Omega_{0} 
		+\varepsilon^{2} \mu^{(2)}\right),
		\label{eq:om2tot}
		\\
		\bar{\omega}_{2} &= \bar{\omega}_{0}
		\left(1+\epsilon\frac{\alpha-2}{\alpha+2}\Omega_{0} 
		+\varepsilon^{2}\bar{\mu}^{(2)} \right).
		\end{align}
		\label{eq:freq2}
	\end{subequations}
	We observe that the precision of $\varphi$ in \eqref{eq:phimultsol} is
	of order $\varepsilon^{2}$, but $\omega_{2}$ \eqref{eq:om2tot} is allowed 
	to contain terms of higher order.
	This is possible since these higher order terms arise when fixing 
	terms of order $\varepsilon$ in the multiple-scale expansion,
	according to the prescriptions in Section \ref{sec:method}.
	In this sense these are not \emph{bona fide} higher order terms, but
	just higher order \emph{corrections} to the $\BigO{1}$ term $\omega_{0}$.
	Now, due to the fact that $\varphi$ is an angle variable, we have that
	at the second order the condition for the orbit to be closed is
	that $\omega_{2}/\bar{\omega}_{2}$ must stay rational.
	Expanding in Taylor series we can write:
	\begin{equation}
	\Delta^{(2)}=  \frac{\omega_{2}}{\bar{\omega}_{2}} =
	\frac{1}{\sqrt{\alpha+2}} 
	+ \frac{\varepsilon^{2}}{24}
	\left(\frac{\rho_{0}}{r_{0}}\right)^{2} 
	\frac{(\alpha+1) (\alpha-2)}{\sqrt{\alpha+2}}
	+ \BigO{\varepsilon^{3}}.
	\label{eq:ratio2nd}
	\end{equation}
	This implies that in addition to the condition \eqref{eq:alphacond}
	we must have either $\alpha=-1$ or $\alpha=2$. 
	Otherwise in the ratio $\omega_{2}/\bar{\omega}_{2}$ there will be
	terms depending on the initial conditions, which can clearly be
	purely irrational.
	We note that if $\alpha=2$, then $k=\kappa^{2}$ and we have 
	$\omega_{2}/\bar{\omega}_{2} = 1/2$ whereas if $\alpha=-1$, then $k=-\kappa^{2}$
	and we have $\omega_{2}/\bar{\omega}_{2} = 1$.
	This means that in these two cases the proportionality is exact up
	to the second order.
	This concludes the proof of the Theorem, since it is known 
	\cite{LeviCivita1922vol1,Whittaker,Moulton2012} that in the two mentioned cases
	the bounded orbits are closed.
	For instance if $\alpha=2$ and $k=\kappa^{2}$ the orbits are given
	by the following expression:
	\begin{equation}
	r^2(\varphi) = \dfrac{\ell^2/E}{1+ \sqrt{1-2 (\kappa \ell/ E)^2}\cos\left[2(\varphi-\varphi_0)\right]} ,
	\label{eq:orbitho}
	\end{equation}
	where $E$ is the total mechanical energy of the system.
	Equation \eqref{eq:orbitho} means that the radius vector describes ellipses  centered at the origin.
	On the other hand, if $\alpha=-1$ and $k=-\kappa^{2}$, the orbits  
	are given by the following expression:
	\begin{equation}
	r(\varphi) = \dfrac{\ell^2/\kappa^2}{1+ \sqrt{1+2  E \ell^2/\kappa^4}\cos(\varphi-\varphi_0)} \, .
	\label{eq:orbitkep}
	\end{equation}
	Equation \eqref{eq:orbitkep} means that, for bounded motion, 
	the radius vector describes ellipses where the origin is one of the foci.
	In both cases, the time dependence is recovered by integrating
	the angular momentum first integral \eqref{eq:angmom}:
	\begin{equation}
	t - t_{0} = \frac{1}{\ell}
	\int_{\varphi_{0}}^{\varphi} r^{2}\left( \tilde{\varphi} \right)\ud \tilde{\varphi}.
	\label{eq:angmomint}
	\end{equation}
	This concludes the proof of the Theorem.
\end{proof}
\noindent 
In Figure \ref{fig:orbits} we show some examples of circular orbits and 
their deformations in the cases highlighted during the proof of 
Bertrand's Theorem \ref{thm:bertrand}.

\begin{figure}[hbt]
	\centering    
	\subfloat[$\alpha=2$: the circular orbit (blue) is stable,
	and the perturbations are periodic ellipses (orange
	and green).
	Exact plot using formula \eqref{eq:orbitho}.]{
		\includegraphics[width=0.47\textwidth]{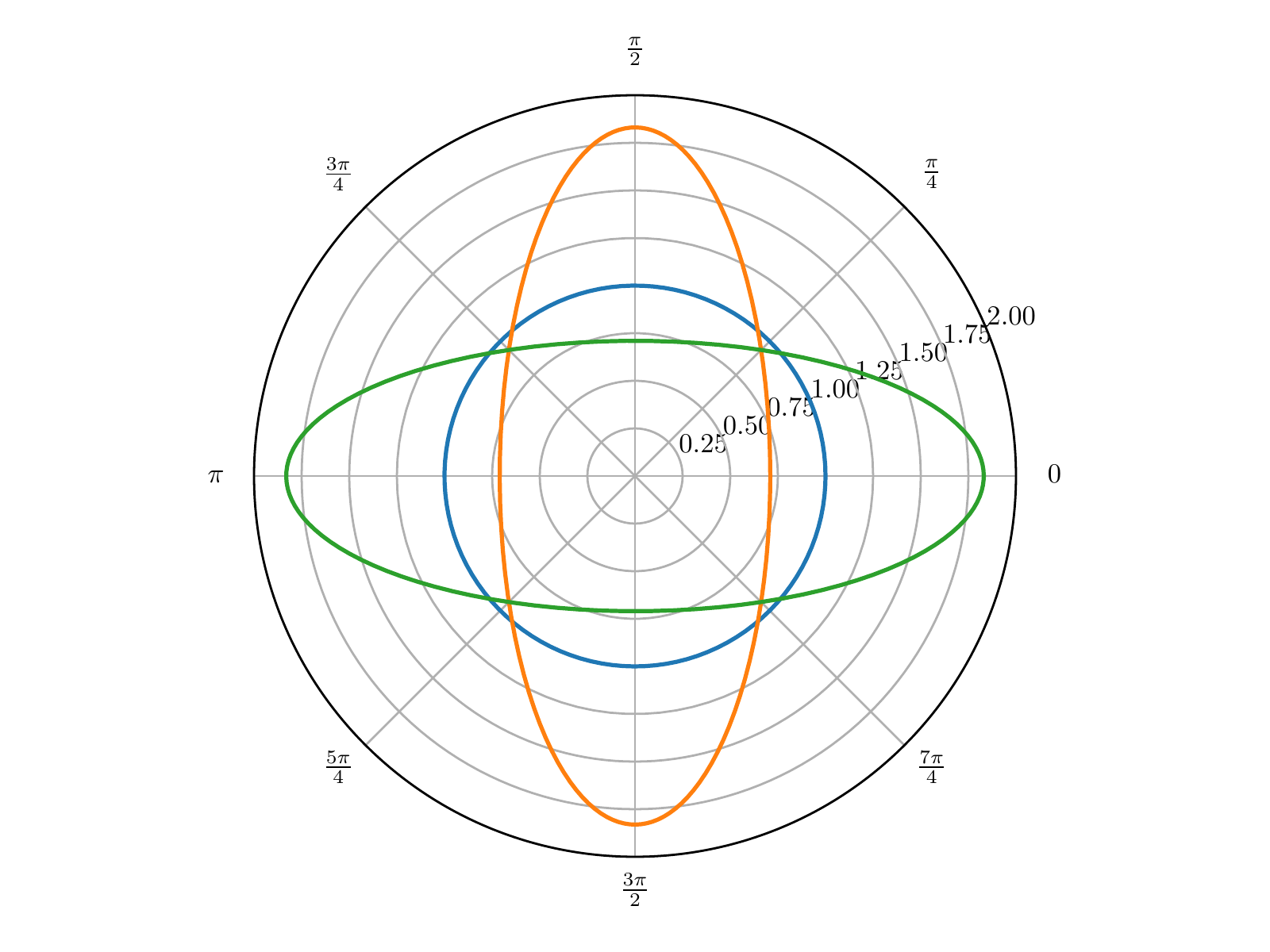}
	}
	\hfill
	\subfloat[][$\alpha=-1$: the circular orbit (blue) is stable,
	and small perturbations are periodic ellipses (orange).
	For high values of the perturbation the orbit
	opens up becoming a parabola (green) or a hyperbola (red).
	Exact plot using formula \eqref{eq:orbitkep}.]{
		\includegraphics[width=0.47\textwidth]{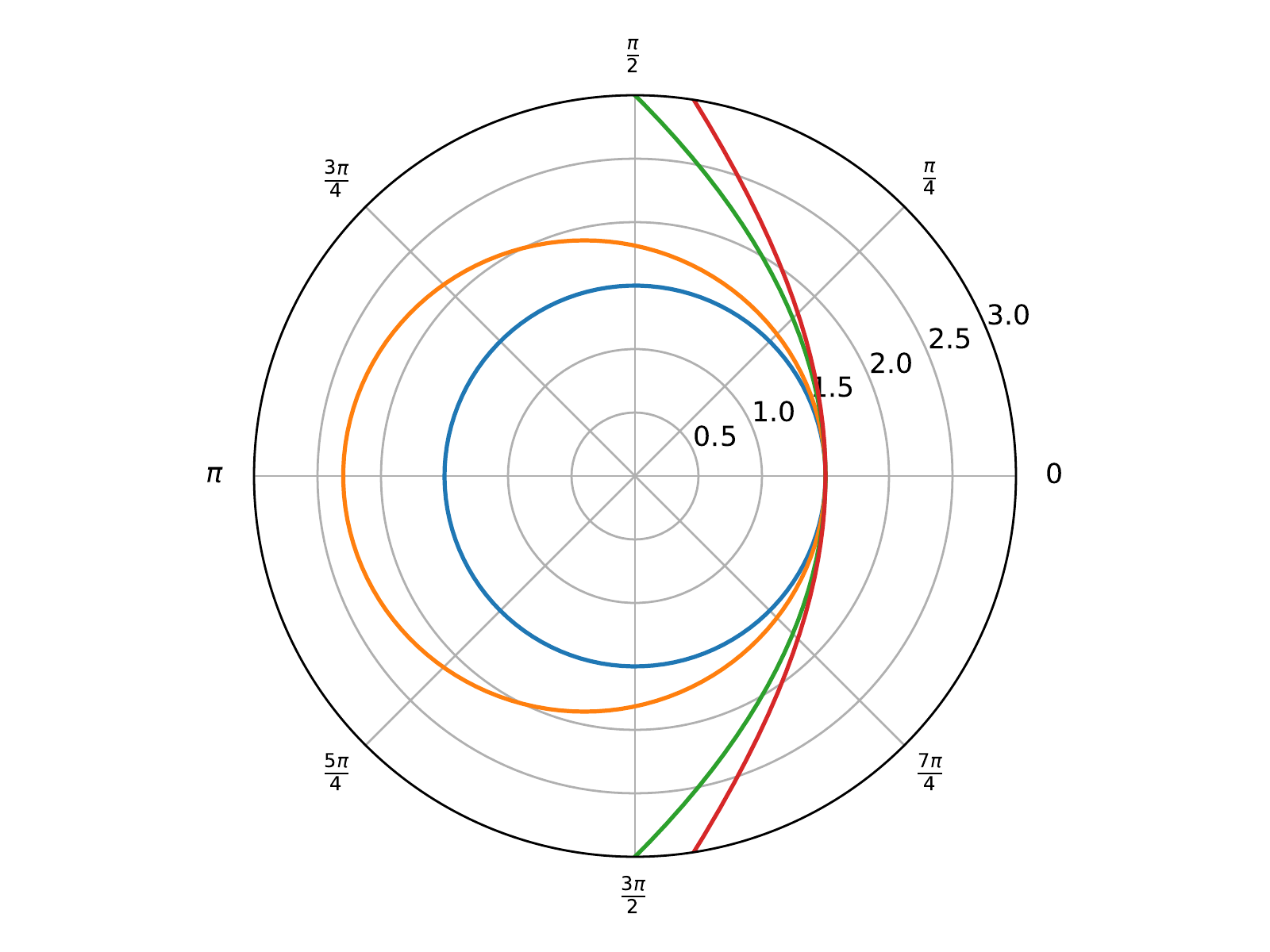}
	}
	\\
	\subfloat[][$\alpha=-2$: the circular orbit (blue) is unstable.
	Small perturbations are non-periodic curves
	known as \emph{Cotes' spirals} (red and green) or \emph{inverse
		spirals} (orange).
	See Appendix \ref{app:orbits}.]{
		\includegraphics[width=0.47\textwidth]{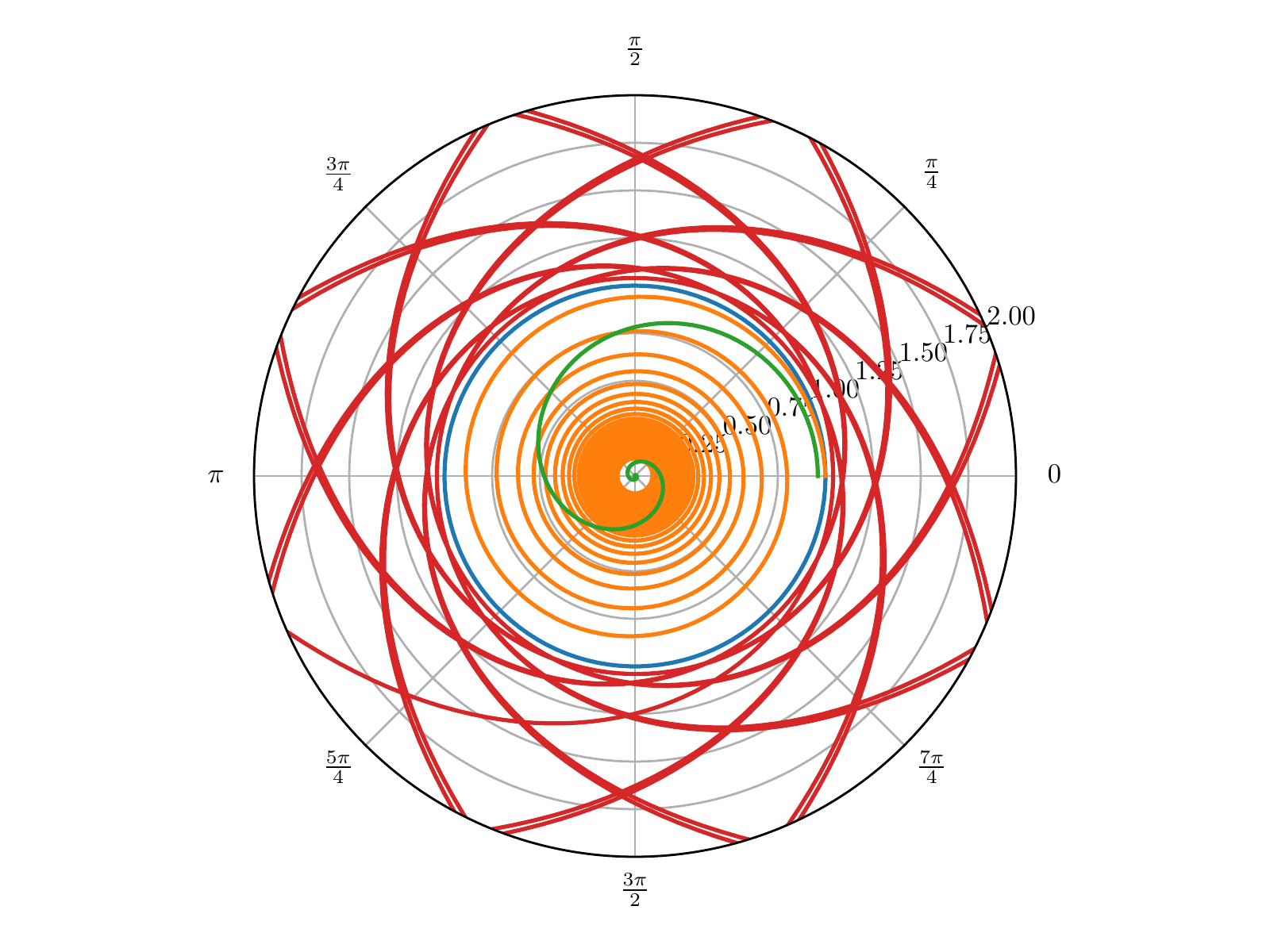}
	}
	\hfill
	\subfloat[][$\alpha=6$: the circular orbit (blue) is stable.
	Small perturbations are non-periodic curves
	defined by an elliptic integral (orange).
	See Appendix \ref{app:orbits}.]{
		\includegraphics[width=0.47\textwidth]{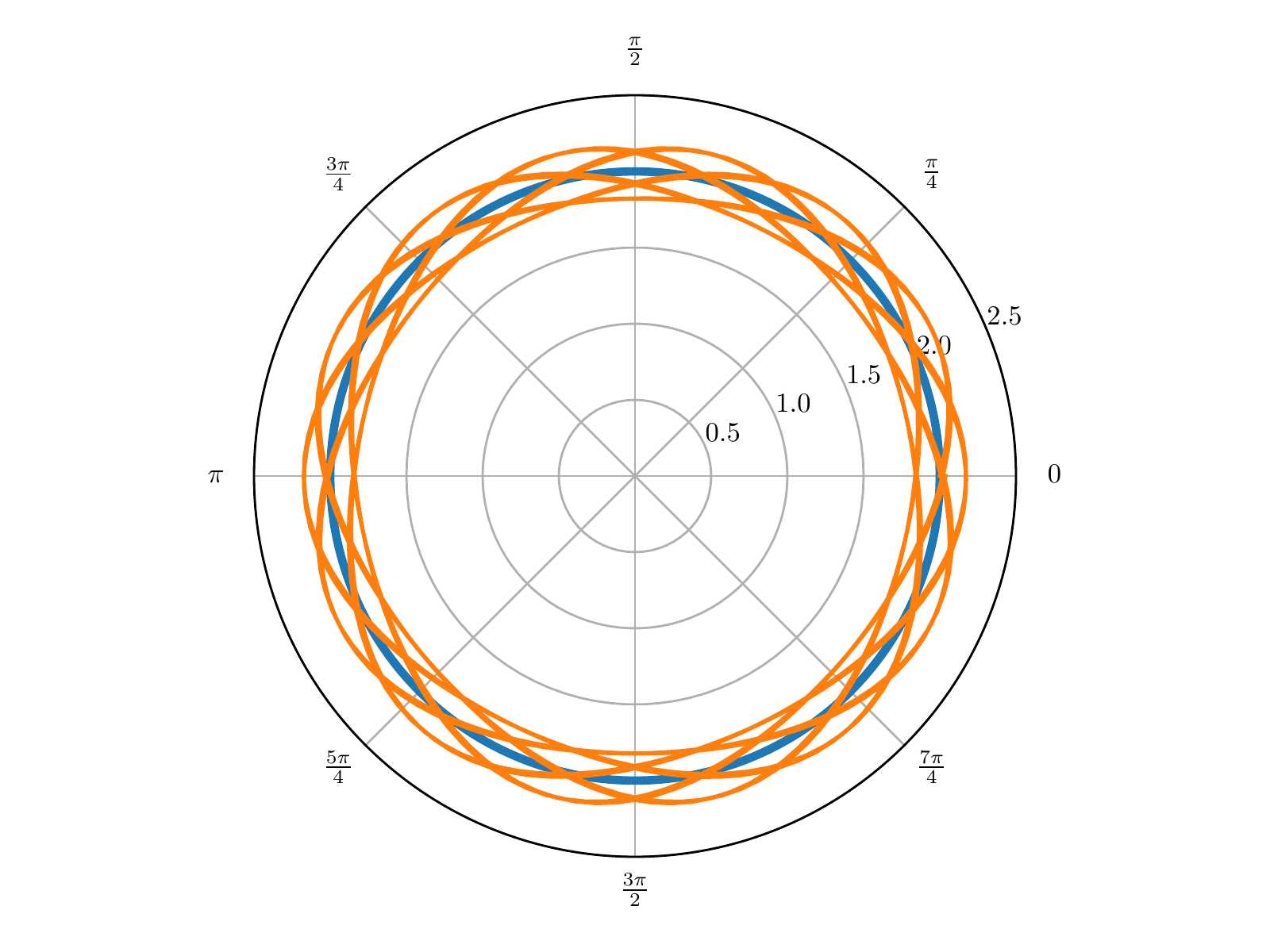}
	}
	\caption{Circular orbits and their deformations
		for various values of $\alpha$.}
	\label{fig:orbits}
\end{figure}

\begin{remark}
	We remark that the proof of 
	Bertrand's Theorem \ref{thm:bertrand} can be accomplished using the Poincar\'e-Lindstedt method \cite{Zarmi2002}, which
	is a particular case of the multiple scales method.
	We remark, however, that our procedure is more general than the one 
	applied in \cite{Zarmi2002}. 
	In fact, we do not need to resort to the orbit equation and 
	we do not need to use periodicity as a hypothesis, 
	but periodicity arises as a consequence.
\end{remark}

\begin{remark}
	We underline that the proof of Bertrand's Theorem \ref{thm:bertrand}  
	could have been accomplished in another way using 
	\emph{non-inertial reference frames} 
	\cite{LeviCivita1922vol1,Whittaker,Appell1893}.
	In fact, let us assume that our motion in the plane is a restriction
	of a three dimensional motion in $\left( q_{1},q_{2},q_{3} \right)\in\R^{3}$ 
	where along the $q_{3}$-axis there is no dynamics. 
	Then, we can switch to a non-inertial reference
	frame moving around the $q_{3}$-axis with angular velocity
	$\boldsymbol{\Theta}= \left( 0,0,\Theta \right)$.
	This corresponds to the following coordinate transformation
	\cite{LeviCivita1922vol1,Whittaker,Appell1893}:
	\begin{equation}
	r = R, \quad \varphi = \Phi + \Theta t,
	\label{eq:xNI}
	\end{equation}
	where we denoted with capital letters the coordinates in the
	non-inertial reference system.
	The transformation \eqref{eq:xNI} brings the Lagrangian \eqref{eq:lrp} 
	into the following form:
	\begin{equation}
	L_{k,\alpha,\Theta} = \half \left( \dot{R}^{2} + R^{2}\dot{\Phi}^{2} \right)
	+\Theta R^{2} \dot{\Phi}
	+\Theta^{2} R^{2}- k R^{\alpha}.
	\label{eq:lrpni}
	\end{equation}
	It is easy to prove that the Euler-Lagrange equations 
	corresponding to \eqref{eq:lrpni}, differently from those
	of \eqref{eq:lrp}, possess the following equilibrium solutions:
	\begin{equation}
	R_{0} = \left(\frac{\Theta^2}{k\alpha}\right)^{1/\left(\alpha-2\right)},
	\label{equini}
	\end{equation}
	and $\Phi_{0}$ can be arbitrary in $[ 0,2\pi)$.
	Changing the value of $\Theta$ we can change the equilibrium
	point in the variables $\left( R,\Phi \right)$. E.g. choosing
	$\Theta=\sqrt{k\alpha}r_{0}^{\alpha/2-1}$, with $r_{0}\in\R^{+}$
	every positive real number can be chosen to be the equilibrium point.
	Therefore, we obtain that the equilibrium positions of the Euler-Lagrange
	equations corresponding to \eqref{eq:lrpni} are given by
	$\left( r_{0},\Phi_{0} \right)$ with fixed $r_{0}\in\R^{+}$ and for all
	$\Phi_{0} \in [0,2\pi)$. These equilibrium points then lie on
	circles centered at $R=0$ of radius $r_{0}$. 
	Adjusting the angular velocity  $\Theta$ of the non-inertial reference frame we can change the radius of the circle. 
	This means that these equilibrium points form a two-parameter
	family of solutions.
	That is, we have a one-to-one correspondence between the equilibrium
	positions of the Euler-Lagrange equations corresponding to
	\eqref{eq:lrpni} and the circular orbits \eqref{eq:circorb}.
	At this point one just needs to perform the usual analysis with
	the derived equilibrium positions and with the Euler-Lagrange equations
	corresponding to \eqref{eq:lrpni}.
	Inverting the point transformation \eqref{eq:xNI} we then obtain
	the result in the original inertial reference frame.
	These reasoning underlines what is well-known in Celestial
	Mechanics \cite{Moulton2012}, i.e. that using conveniently non-inertial
	reference frames, in some cases it is possible to map particular
	solutions to equilibrium points. 
	This fact can be also used, in principle, in other cases where 
	no equilibrium solutions exist, but it is possible to find simple 
	classes of solutions.
\end{remark}

\section{Conclusions and outlook}
\label{sec:conclusion}

In this paper we presented a simple method which allows to establish 
if a Hamiltonian system may be or not maximally superintegrable.
The main benefit of this approach is the fact that it is algorithmic
and relies on the well-established method of the multiple scales expansion.
As stated in Section \ref{sec:method} the multiple scales analysis
proved to be very useful in the theory of integrable systems in
infinite dimensions.
We believe that the procedure presented in this paper can be thought
as its \emph{finite dimensional analog}.
Indeed, it has been observed that integrable systems in infinite
dimensions are, as a matter of fact, maximally superintegrable.
The generalized symmetries of these equations form infinite 
dimensional non-Abelian algebras (the Orlov-Shulman symmetries) 
with infinite dimensional Abelian subalgebras of commuting
flows \cite{Orlov1985,Orlov1986,Orlov1997}.
Most notably, the multiple scales method has never been
systematically applied to finite dimensional systems. 

\noindent We applied this method to five relevant examples in order to understand its features and its limits:
\begin{enumerate}
    \item The generalized H\'enon-Heiles system \eqref{eq:HHH}.
    \item The anisotropic caged oscillator \eqref{eq:cagedoscillator}.
    \item The TTW system \eqref{eq:ttwconf}.
    \item The Drach system \eqref{eq:Hdrach}.
    \item The central force Hamiltonian \eqref{eq:hrp}.
\end{enumerate}
From the example of the generalized H\'enon-Heiles system \eqref{eq:HHH}
we explicitly showed that our approach, based only on equilibrium
positions and perturbative expansions, cannot distinguish between
chaoticity and integrability, but only between non maximal superintegrability
and maximal superintegrability. From the example of the anisotropic caged oscillator \eqref{eq:cagedoscillator}
we showed that our test can detect maximal superintegrability
where the Painlev\'e test fails.
This shows that our maximal superintegrability test is not
necessarily stronger than the Painlev\'e property as the example
of the generalized H\'enon-Heiles system could have suggested.
This implies that the application of these two different tests,
both based on the analysis of the property of the solution of
a mechanical system, can be thought as complementary. 
From the example of the TTW system \eqref{eq:ttwconf} we
showed how the results of our technique are able to suggest
the superintegrability of this well-known model.
Furthermore, with the Drach system \eqref{eq:Hdrach} we gave
a practical example of the method by showing that
this system is, in general, not (maximally) superintegrable.
Finally, with the last example  we gave some hints on how this approach could be applied
to classify maximally superintegrable systems  by presenting a direct proof of Bertrand's Theorem \ref{thm:bertrand} . During the proof we gave a slight generalization of the algorithm presented in Section \ref{sec:method}. We remark that this generalization can be applied every time a parametric family of bounded and closed solutions is known. Moreover, we observe that during the proof we have recovered the following Lemma which is interesting on its own:
\begin{lemma}
    Assume we are given the Hamiltonian in the Euclidean plane
    $H_{k,\alpha}$ of equation \eqref{eq:hrp}.
    Then, if $k\alpha>0$ and $\alpha>-2$  circular orbits \eqref{eq:circorb}
    exist and are stable.
    \label{lemm:orbstab}
\end{lemma}
\noindent Amongst the systems admitting stable circular orbits the
isotropic harmonic oscillator and the Kepler system
occupy a special place, having their finite orbits closed
and periodic.
This emphasizes how in the \emph{classification problem}
a perturbative approach like ours phase out the generic cases,
leaving only those with interesting properties.
The remaining cases can then be treated directly, and their
properties are usually manifest.
Moreover, we could interpret the fact that the
isotropic harmonic oscillator and the Kepler system
possess special properties, beneath a huge amount of similar
systems, as a finite dimensional analog of the
fact that almost every evolution equation possesses a $1$-soliton
solution, whereas $N$-soliton solutions are a unique feature
of integrable equations.
Furthermore, the many non-integrable equations possessing
also a $2$-soliton solution can be seen as the analog of the potentials
satisfying the hypothesis of Lemma \ref{lemm:orbstab}.
Indeed, the true distinction between the integrable and the non-integrable
cases arises when the existence of a $3$-soliton solution is required.
This is because the interaction of three solitons
of a non-integrable evolution equations becomes destructive.
In the same way in our finite dimensional case the closeness
of circular orbits is not preserved when we arrive at the third
time scale.
We mention that the existence of $3$-soliton solutions
have been used to classify soliton equations in 
\cite{HietarintaBilinI,HietarintaBilinII,HietarintaBilinIII,
HietarintaBilinIV,HietarintaBilinRev} 
using the so-called Hirota bilinear method \cite{Hirota1971}.
We also note that soliton solutions of integrable equations
are usually stable with respect to some norm in an appropriate
function space \cite{ZakharovShabat1973,MaddocksSachs1993}.
Intuitively the stability of the $N$-soliton solution
for an integrable equation can be described as follows:
given  ``sufficiently regular'' and rapidly decreasing initial data it is possible to 
construct the inverse scattering.
The evolution of the spectral data then implies that the given
initial data will evolve producing a certain number of solitons and
a radiative background.
The area of the radiative background is ``small'' compared to that
of the solitions.
Then, as time grows, the radiative background decays whereas the
solitons will keep their form unaltered, since in the collisions only the relative shifts will change.
This comment underlines another possible analogy between the 
finite and the infinite dimensional case.

To conclude we have that, in general, maximal superintegrability can 
be proved with the usual approaches, see \cite{MillerPostWinternitz2013R}
and references therein, but the procedure we presented can be
used as a \emph{sieve test} for maximally superintegrable systems.
Based on this evidence we think that the application of
this approach in a more general setting, where instead of
arbitrary constants, namely $\alpha$ and $k$ in Bertrand's
Theorem \ref{thm:bertrand}, we have arbitrary functions may be a new
way to classify families of maximally superintegrable systems.
We reserve the application of this method in this more general
setting to future works.

\section*{Acknowledgments}

We thank prof. M. C. Nucci, prof. R. I. Yamilov, Dr. D. Riglioni 
and Dr. F. Zullo for
interesting and helpful discussions during the preparation of this paper. 

GG is supported by the Australian Research Council through 
an Australian Laureate Fellowship grant FL120100094.

\appendix

\section{Orbits of \eqref{eq:hrp} if $\boldsymbol{\alpha =-2}$ or $\boldsymbol{\alpha=6}$}
\label{app:orbits}

For a system with radial symmetry the orbits
can be computed in general using the so-called
\emph{orbit equation} \cite{Whittaker,LeviCivita1922vol1,Appell1893,Moulton2012}:
\begin{equation}
    \dv[2]{u}{\varphi}+u=\frac{k\alpha}{\ell^{2}}u^{-\alpha-1},
    \label{eq:orbeq}
\end{equation}
where $u=1/r$ and $\varphi$ is the new dependent variable
defined through the differential substitution obtained
from \eqref{eq:angmom}:
\begin{equation}
    \ud t  =\frac{r^{2}}{\ell} \ud \varphi.
    \label{eq:dtdphi}
\end{equation}
The differential substitution \eqref{eq:dtdphi} is possible
since $\ell$ is constant along the solutions of \eqref{eq:lrpeq}.
Performing the differential substitution \eqref{eq:dtdphi}
we have that \eqref{eq:phieq} is satisfied identically
and \eqref{eq:req} is transformed into:
\begin{equation}
    \frac{\ell}{r^{2}}\dv{\varphi}\left( \frac{\ell}{r^{2}}\dv{r}{\varphi} \right)-
    \frac{\ell^{2}}{r^{3}} = -k\alpha r^{\alpha-1}.
    \label{eq:orb0}
\end{equation}
Applying the transformation $r=1/u$ equation \eqref{eq:orbeq} follows.
Here, we will concentrate on the two cases $\alpha=-2$ and $\alpha=6$.

\begin{description}
    \item[Case $\boldsymbol{\alpha=-2}$.]
	If $\alpha=-2$ then the orbit equation \eqref{eq:orbeq}
	is \emph{linear}.
	Due to the Remark made during the proof of Bertrand's Theorem
	we impose $k=-\kappa^{2}/2$.
	Defining:
	\begin{equation}
	\lambda \doteq 1 - \frac{\kappa^2}{\ell^2} \, ,
	\label{eq:lam2}
	\end{equation}
	the solutions in terms of $r$ are 
	\cite{Whittaker,LeviCivita1922vol1,Appell1893,Moulton2012}:
	\begin{equation}
	r =
	\begin{cases}
	\displaystyle
	\frac{r_{0}\ell}{\ell\cos\left[\sqrt{\lambda}\left( \varphi-\varphi_{0} \right)\right]
            -r_{0}\dot{r}_{0}\sin\left[\sqrt{\lambda}\left( \varphi-\varphi_{0} \right)\right]},
	&
	\text{if $\lambda>0$},
	\\
	\displaystyle
	\frac{r_{0}\ell}{\ell\cosh\left[\sqrt{-\lambda}\left( \varphi-\varphi_{0} \right)\right]
            -r_{0}\dot{r}_{0}\sinh\left[\sqrt{-\lambda}\left( \varphi-\varphi_{0} \right)\right]},
	&
	\text{if $\lambda<0$},
	\\
	\displaystyle
        \frac{r_{0}\ell}{\ell-r_{0}\dot{r}_{0}\left( \varphi-\varphi_{0} \right)}.
	&
	\text{if $\lambda=0$}.
	\end{cases}
	\label{eq:orb2}
	\end{equation}
	We denoted $r(t=0)=r_{0}$ and $\dot{r}(t=0)=\dot{r}_{0}$.
	
	The curves described by the radius vector in \eqref{eq:orb2}
	if $\lambda>0$ and $\lambda<0$ are called \emph{Cotes' spirals}
	whereas the last one, i.e. $\lambda=0$, is called \emph{reciprocal spiral} \cite{Whittaker}.
	Circular orbits arise in the latter case when $\dot{r}_{0}=0$.
	A circular orbit when perturbed with a slight radial velocity
	becomes a reciprocal spiral. 
	On the contrary when perturbed with a slight \emph{positive} 
	angular velocity it becomes a trigonometric Cotes' spiral.
	Finally when perturbed with a slight \emph{negative} angular
	velocity it becomes an hyperbolic Cotes' spiral.
	
	In Figure \ref{fig:orbits} these cases are illustrated explicitly
	using the appropriate exact values of $r$ from \eqref{eq:orb2},
	with $\kappa=1$.
	\item[Case $\boldsymbol{\alpha=6}$.] 
	When $\alpha=6$ the orbit equation \eqref{eq:orbeq} becomes:
	\begin{equation}
	\dv[2]{u}{\varphi}+u=\frac{\kappa^{2}}{\ell^{2}}u^{-7}.
	\label{eq:orbeq6}
	\end{equation}
	Due to the remark made during the proof of Bertrand's Theorem
	we imposed $k=\kappa^{2}/6$.
	Multiplying by $\dv*{u}{\varphi}$ and integrating we
	transform \eqref{eq:orbeq6} into the following
	first order equation:
	\begin{equation}
	\half\left( \dv{u}{\varphi} \right)^2+\half u^2+\frac{\kappa^{2}}{6\ell^2u^6}=C,
	\label{eq:orb6red}
	\end{equation}
	where $C$ is a constant of integration.
	Performing the transformation $u = \sqrt{v}$ we obtain:
	\begin{equation}
	\frac{1}{8} \left(\dv{v}{\varphi}\right)^2
	+\half v^2+\frac{\kappa^2}{6\ell^2 v^2}= Cv.
	\label{eq:orb6redsq}
	\end{equation}
	Equation \eqref{eq:orb6redsq} means that $v$ is defined by an
	\emph{elliptic function} obtained through the inversion of the
	\emph{elliptic integral} \cite{WhittakerWatson1927}:
	\begin{equation} 
	\varphi -\varphi_{0}
	=\pm \frac{1}{2\sqrt{2}} \int^{v}\left(C v-\half v^2-\frac{\kappa^2}{6\ell^2 v^2}\right)^{-1/2}\ud v.
	\label{eq:orb6sol}
	\end{equation}
	This describes completely the orbits in the case $\alpha=6$.
	
	The numerical evaluation of the integral \eqref{eq:orb6sol} is
	particularly stiff.
	In addition, since we are interested in circular orbits it is
	even more difficult, being circular orbits degenerate
	solutions where the dependence on $\varphi$ is suppressed.
	For this reason to produce the orbit in Figure \ref{fig:orbits}
	we resorted to the numerical integration of the following second-order initial value problem for $v$:
	\begin{equation}
	\left\{
	\begin{aligned}
	&\half \dv[2]{v}{\varphi}-\frac{1}{4v} \left(\dv{v}{\varphi}\right)^2
	+v=\frac{\kappa^2}{\ell^{2}v^3},
	\\
	&v|_{\varphi=\varphi_{0}}=\frac{1}{r_{0}^{2}}, 
	\quad 
	\left.\dv{v}{\varphi}\right|_{\varphi=\varphi_{0}}=-\frac{2\dot{r}_{0}}{r_{0}\ell}.
	\end{aligned}
	\right.
	\label{eq:orb6vsecond}
	\end{equation}
	Equation \eqref{eq:orb6vsecond} is obtained from \eqref{eq:orbeq6} 
	by performing the transformation $u=\sqrt{v}$.
	We then used the \texttt{odeint} integrator from \texttt{scipy} \cite{scipy}
	with a regular mesh of $N=2^{25}$ points in the interval 
	$[0,20\pi]$ with $\kappa=1$.
\end{description}

\bibliographystyle{plain}
\bibliography{bibliography.bib}

\begin{thebibliography}{10}

\bibitem{AbendaFedorov2000}
S.~Abenda and Yu. Fedorov.
\newblock On the weak {K}owalevski--{P}ainlev\'e property for hyperelliptically
  separable systems.
\newblock {\em Acta Appl. Math.}, 60:137--178, 2000.

\bibitem{AbendaFedorov2001}
S.~Abenda and Yu. Fedorov.
\newblock Complex angle variables for constrained integrable hamiltonian
  systems.
\newblock {\em J. Nonlinear Math. Phys.}, 8, Supplement 1:1--4, 2001.

\bibitem{Appell1893}
P.~Appel.
\newblock {\em Trait\'e de M\'ecanique Rationelle}, volume~1.
\newblock Gauthier-Villars et fils, Paris, 1893.

\bibitem{Arnold1967}
V.~I. Arnol'd.
\newblock On a theorem of {L}iouville concerning integrable problems in
  dynamics.
\newblock {\em Amer. Math. Soc. Trasl.}, 61:292--296, 1967.

\bibitem{Arnold1997}
V.~I. Arnol'd.
\newblock {\em {Mathematical Methods of Classical Mechanics}}, volume~60 of
  {\em Graduate Texts in Mathematics}.
\newblock Springer-Verlag, Berlin, 2nd edition, 1997.

\bibitem{BenderOrszag}
C.~M. Bender and S.~A. Orszag.
\newblock {\em Advanced mathematical methods for scientists and engineers}.
\newblock McGraw-Hill, 1978.

\bibitem{Bertrand1873}
J.~Bertrand.
\newblock Th\'eor\`eme relatif au mouvement d'un point attir\'e vers un centre
  fixe.
\newblock {\em C. R. Acad. Sci.}, 77:849--853, 1873.

\bibitem{BoccalettiPucacco2004}
D.~Boccaletti and G.~Pucacco.
\newblock {\em Theory of Orbits}.
\newblock Springer-Verlag, Berlin, 2004.

\bibitem{BountisSegurVivaldi1982}
T.~Bountis, H.~Segur, and F.~Vivaldi.
\newblock Integrable {H}amiltonian systems and the {P}ainlev\'e property.
\newblock {\em Phys. Rev. A}, 25:1257--1264, 1982.

\bibitem{Braun1993}
M.~Braun.
\newblock {\em Differential Equations and Their Applications: An Introduction
  to Applied Mathematics}.
\newblock Applied mathematical sciences. Springer, 1993.

\bibitem{Calogero1971}
F.~Calogero.
\newblock Solution of the one-dimensional {N}-body problem with quadratic
  and/or inversely quadratic pair potentials.
\newblock {\em J. Math. Phys.}, 12:419--436, 1971.

\bibitem{Calogero1991}
F.~Calogero.
\newblock Why are certain nonlinear {PDEs} both widely applicable and
  integrable?
\newblock In V.~E. Zakharov, editor, {\em What is integrability?} Springer,
  Berlin-Heidelberg, 1991.

\bibitem{Calogero1971erratum}
F.~Calogero.
\newblock Erratum on ``solution of the one-dimensional {N}-body problem with
  quadratic and/or inversely quadratic pair potentials''.
\newblock {\em J. Math. Phys.}, 37:3646, 1996.

\bibitem{Calogero2008book}
F.~Calogero.
\newblock {\em Isochronous Systems}.
\newblock OUP Oxford, Oxford, 2008.

\bibitem{CalogeroEchkhaus1987}
F.~Calogero and W.~Echkhaus.
\newblock Nonlinear evolution equations, rescalings, model {PDEs} and their
  integrability, {I}.
\newblock {\em Inv. Probl.}, 3:229--262, 1987.

\bibitem{CalogeroEchkhaus1988}
F.~Calogero and W.~Echkhaus.
\newblock Nonlinear evolution equations, rescalings, model {PDEs} and their
  integrability, {II}.
\newblock {\em Inv. Probl.}, 4:11--13, 1988.

\bibitem{ChangTaborWeiss1982}
Y.~F. Chang, M.~Tabor, and J.~Weiss.
\newblock Analytic structure of the {H}\'enon-{H}eiles {H}amiltonian in
  integrable and nonintegrable regimes.
\newblock {\em J. Math. Phys.}, 23:531--538, 1982.

\bibitem{ColeKevorkian1963}
J.~D. Cole and J.~Kevorkian.
\newblock Uniformly valid asymptotic approximations for certain differential
  equations.
\newblock In J.~P. LaSalle and S.~Lefschetz, editors, {\em Nonlinear
  differential equations and Nonlinear Mechanics}. Academic Press, New York,
  1963.

\bibitem{painleveproperty1999}
R.~Conte, editor.
\newblock {\em The Painlev\'e property}, CRM Series in Mathematical Physics,
  Berlin, 1999. Springer-Verlag.

\bibitem{DabulSlodowyDabul1993}
J.~Daboul, P.~Slodowy, and C.~Daboul.
\newblock The hydrogen algebra as centerless twisted {K}ac-{M}oody algebra.
\newblock {\em Phys. Lett. B}, 317:321--328, 1993.

\bibitem{DorizziGrammaticosRamani1983}
B.~Dorizzi, B.~Grammaticos, and A.~Ramani.
\newblock A new class of integrable systems.
\newblock {\em J. Math. Phys.}, 24:2282--2288, 1983.

\bibitem{Drach1935}
J.~Drach.
\newblock Sur l'int\'egration logique des \'equations de la dynamique \`a deux
  variables: forces conservatives. {I}nt\'egrales cubiques. {M}ouvements dans
  le plan.
\newblock {\em C. R. Acad. Sci.}, 200:22--26, 1935.

\bibitem{Dubrovin1984}
B.~A. Dubrovin, A.~T. Fomenko, and F.~T. Novikov.
\newblock {\em Modern {G}eometry - {M}ethods and {A}pplications: {P}art {I}.
  {T}he {G}eometry of {S}urfaces, {T}rasformations {G}roups and {F}ields}.
\newblock Springer-Verlag, New York, {II} edition, 1992.

\bibitem{Evans1990}
N.~W. Evans.
\newblock Super-integrability of the {W}internitz system.
\newblock {\em Phys. Lett. A}, 147(8):483 -- 486, 1990.

\bibitem{Evans2008}
N.~W. Evans and P.~E. Verrier.
\newblock Superintegrability of the caged anisotropic oscillator.
\newblock {\em J. Math. Phys.}, 49(9):092902, 2008.

\bibitem{Fordy1991}
A.~Fordy.
\newblock The {H}\'enon-{H}eiles system revisited.
\newblock {\em Physica D}, 52:204--210, 1991.

\bibitem{Fris1965}
J.~Fri\v{s}, V.~Mandrosov, Ya.~A. Smorodinski, M.~Uhl\'i\v{r}, and
  P.~Winternitz.
\newblock On higher symmetries in {Quantum Mechanics}.
\newblock {\em Phys. Lett.}, 13(3), 1965.

\bibitem{Fuchs1906}
R.~Fuchs.
\newblock Sur quelques \'equations diff\'erentielles lin\'eaires du second
  ordre.
\newblock {\em Comptes Rendus}, 141:555--558, 1906.

\bibitem{Gambier1910}
B.~Gambier.
\newblock Sur les \'equations diff\'erentielles du second ordre et du premier
  degré dont l'intégrale g\'en\'erale est \`a points critiques fixes.
\newblock {\em Acta Math.}, 33:1--55, 1910.

\bibitem{GrammaticosDorizziPadjen1982}
B.~Grammaticos, B.~Dorizzi, and R.~Padjen.
\newblock Painlev\'e property and integrals of motion for the
  {H}\'enon-{H}eiles system.
\newblock {\em Phys. Lett. A}, 89:111--113, 1982.

\bibitem{Gutzwiller1990}
M.~C. Gutzwiller.
\newblock {\em Chaos in {C}lassical and {Q}uantum {M}echanics}, volume~1 of
  {\em Interdisciplinary Applied Mathematics}.
\newblock Springer-Verlag, New York, 1990.

\bibitem{HenonHeiles1964}
M.~H\'enon and C.~Heiles.
\newblock The applicability of the third integral of motion: some numerical
  experiments.
\newblock {\em Astron. J.}, 69:73--79, 1964.

\bibitem{HietarintaBilinI}
J.~Hietarinta.
\newblock A search of bilinear equations passing {H}irota's three-soliton
  condition. {I.} {{KdV}‐type} bilinear equations.
\newblock {\em J. Math. Phys.}, 28:1732--1742, 1987.

\bibitem{HietarintaBilinII}
J.~Hietarinta.
\newblock A search of bilinear equations passing {H}irota's three-soliton
  condition. {II.} {{mKdV}-type} bilinear equations.
\newblock {\em J. Math. Phys.}, 28:2094--2101, 1987.

\bibitem{HietarintaBilinIII}
J.~Hietarinta.
\newblock A search of bilinear equations passing {H}irota's three-soliton
  condition. {III.} {Sine-Gordon}‐type bilinear equations.
\newblock {\em J. Math. Phys.}, 28:2586--2592, 1987.

\bibitem{HietarintaBilinIV}
J.~Hietarinta.
\newblock A search of bilinear equations passing {H}irota's three-soliton
  condition. {IV.} {C}omplex bilinear equations.
\newblock {\em J. Math. Phys.}, 29:628--635, 1988.

\bibitem{HietarintaBilinRev}
J.~Hietarinta.
\newblock Recent results from the search for bilinear equations having
  three-soliton solutions.
\newblock In {\em Nonlinear Evolution Equations: Integrability and Spectral
  Methods}, pages 307--317. Manchester University Press, Manchester, 1989.

\bibitem{Hirota1971}
R.~Hirota.
\newblock Exact solution of the {K}orteweg-de {V}ries equation for multiple
  collisions of solitons.
\newblock {\em Phys. Rev. Lett.}, 27:1192--1194, 1971.

\bibitem{Holmes}
M.~H. Holmes.
\newblock {\em Introduction to Perturbation Methods}.
\newblock Springer, 2013.

\bibitem{Hone2009}
A.~N.~W. Hone.
\newblock Painlev\'e tests, singularity structure and integrability.
\newblock In A.~V. Mikhailov, editor, {\em Integrability}, Lecture Notes in
  Physics, chapter~7, pages 245--277. Springer-Verlag, Berlin, 2009.

\bibitem{scipy}
E.~Jones, T.~Oliphant, P.~Peterson, et~al.
\newblock {SciPy}: Open source scientific tools for {Python}, 2001--\the\year.
\newblock [Online; accessed \today].

\bibitem{KalninsKress2010}
E.~G. Kalnins, J.~M. Kress, and W.~Jr. Miller.
\newblock Families of classical superintegrable systems.
\newblock {\em J. Phys. A: Math. Theor.}, 43:092001, 2010.

\bibitem{KalninsMillerPogosyan2011}
E.~G. Kalnins, W.~Jr. Miller, and G.~S. Pogosyan.
\newblock Superintegrability and higher order constants for classical and
  quantum systems.
\newblock {\em Phys. At. Nucl.}, 74:914--918, 2011.

\bibitem{KevorkianCole}
J.~Kevorkian and J.~D. Cole.
\newblock {\em Multiple scale and singular perturbation methods}.
\newblock Springer-Verlag, 1996.

\bibitem{Kuzmak1959}
G.~E. Kuzmak.
\newblock Asymptotic solutions of nonlinear second order differential equations
  with variable coefficients.
\newblock {\em J. Appl. Math. Mech. (PMM)}, 23:730--744, 1959.

\bibitem{Lee}
T.~Lee, L.~Leok, and N.~H. McClamroch.
\newblock {\em {G}lobal {F}ormulations of {L}agrangian and {H}amiltonian
  {D}ynamics on {M}anifolds}.
\newblock Interaction of Mechanics and Mathematics. Springer International
  Publishing, Cham, 2017.

\bibitem{LeviCivita1922vol1}
T.~Levi-Civita and U.~Amaldi.
\newblock {\em Lezioni di Meccanica Razionale}, volume~1.
\newblock Zanichelli Editore, Bologna, 1922.

\bibitem{Lindstedt1882}
A.~Lindstedt.
\newblock {\"U}ber die integration einer f\"ur die st\"orungstheorie wichtigen
  differentialgleichung.
\newblock {\em Astron. Nachr.}, 103:211--220, 1882.

\bibitem{Liouville1855}
J.~Liouville.
\newblock Note sur l'int\'{e}gration des \'{e}quations diff\'{e}rentielles de
  la {D}ynamique, pr\'{e}sent\'{e}e au {B}ureau des {L}ongitudes le 29 juin
  1853.
\newblock {\em J. Math. Pures Appl.}, 20:137--138, 1853.

\bibitem{MaddocksSachs1993}
J.~H. Maddocks and R.~L. Sachs.
\newblock On the stability of {KdV} multi-solitons.
\newblock {\em Comm. Pure Appl. Math.}, 46:867--901, 1993.

\bibitem{MillerPostWinternitz2013R}
W.~Jr. Miller, S.~Post, and P.~Winternitz.
\newblock Classical and quantum superintegrability with applications.
\newblock {\em J. Phys. A: Math. Theor.}, 46(42):423001, 2013.

\bibitem{Moser1975}
J.~Moser.
\newblock Three integrable hamiltonian systems connected with isospectral
  deformations.
\newblock {\em Adv. Math.}, 16:197--220, 1975.

\bibitem{Moulton2012}
F.~R. Moulton.
\newblock {\em An Introduction to Celestial Mechanics}.
\newblock Dover Books on Astronomy. Dover Publications, New York, 2nd edition,
  2012.
\newblock Originally published by The Macmillan Company in 1902.

\bibitem{Nayfeh}
A.~H. Nayfeh.
\newblock {\em Perturbation Methods}.
\newblock John Wiley and Sons, 1973.

\bibitem{Nekhoroshev1972}
N.~N. Nekhoroshev.
\newblock Action-angle variables and their generalizations.
\newblock {\em Trans. Moscow Math. Soc.}, 26:180--198, 1972.
\newblock In Russian.

\bibitem{Nucci1996b}
M.~C. Nucci.
\newblock The complete {K}epler group can be derived by {L}ie group analysis.
\newblock {\em Journal of Mathematical Physics}, 37(4):1772--1775, 1996.

\bibitem{Nucci2000}
M.~C. Nucci and P.~G.~L. Leach.
\newblock The determination of nonlocal symmetries by the technique of
  reduction of order.
\newblock {\em Journal of Mathematical Analysis and Applications}, 251(2):871
  -- 884, 2000.

\bibitem{Nucci2001}
M.~C. Nucci and P.~G.~L. Leach.
\newblock The harmony in the {K}epler and related problems.
\newblock {\em J. Math. Phys.}, 42(2):746--764, 2001.

\bibitem{NucciPost2012}
M.~C. Nucci and S.~Post.
\newblock Lie symmetries and superintegrability.
\newblock {\em J. Phys. A: Math. Theor.}, 45(48):482001, 2012.

\bibitem{Orlov1985}
Yu.~A. Orlov and E.~I. Shulman.
\newblock Additional symmetries of the nonlinear {S}chr\"odinger equation.
\newblock {\em Theor. Math. Phys.}, 64:862--866, 1985.

\bibitem{Orlov1986}
Yu.~A. Orlov and E.~I. Shulman.
\newblock Additional symmetries for integrable equations and conformal algebra
  representation.
\newblock {\em Lett. Math. Phys.}, 12:171--179, 1986.

\bibitem{Orlov1997}
Yu.~A. Orlov and P.~Winternitz.
\newblock Algebra of pseudodifferential operators and symmetries of equations
  in the {K}adomtsev-{P}etviashvili hierarchy.
\newblock {\em J. Math. Phys.}, 38:4644--4674, 1997.

\bibitem{Painleve1900}
P.~Painlev\'e.
\newblock M\'emoire sur les \'equations diff\'erentielles dont l'intégrale
  g\'en\'erale est uniforme.
\newblock {\em Bull. Soc. Math. Phys. France}, 28:201--261, 1900.

\bibitem{Painleve1902}
P.~Painlev\'e.
\newblock Sur les \'equations diff\'erentielles du second ordre et d'ordre
  sup\'erieur dont l'int\'egrale g\'en\'erale est uniforme.
\newblock {\em Acta Math.}, 25:1--85, 1902.

\bibitem{Picard1889}
E.~Picard.
\newblock M\'emoire sur la th\'eorie des fonctions alg\'ebriques de deux
  variables.
\newblock {\em J. de Math. pures appl.}, 5:135--319, 1889.

\bibitem{Poincare1886}
H.~Poincar\'e.
\newblock Sur les int\'egrales irr\'eguli\`eres des \'equations lin\'eaires.
\newblock {\em Acta Math.}, 8:295--344, 1886.

\bibitem{PostWinternitz2011}
S.~Post and P.~Winternitz.
\newblock A nonseparable quantum superintegrable system in {2D} real
  {E}uclidean space.
\newblock {\em J. Phys. A: Math. Theor.}, 44:162001, 2011.

\bibitem{Puiseux1850}
V.~A. Puiseux.
\newblock Recherches sur les fonctions alg\'ebriques.
\newblock {\em J. Math. Pures Appl.}, 15:365--480, 1850.

\bibitem{Puiseux1851}
V.~A. Puiseux.
\newblock Nouvelle recherches sur les fonctions alg\'ebriques.
\newblock {\em J. Math. Pures Appl.}, 16:228--240, 1851.

\bibitem{Quesne2010}
C.~Quesne.
\newblock Superintegrability of the {T}remblay--{T}urbiner--{W}internitz
  quantum {H}amiltonian on a plane for odd $k$.
\newblock {\em J. Phys. A: Math. Theor.}, 43:082001, 2010.

\bibitem{RamaniDorizziGrammaticos1982}
A.~Ramani, B.~Dorizzi, and B.~Grammaticos.
\newblock Painlev\'e conjecture revisited.
\newblock {\em Phys. Rev. Lett.}, 49:1539--1541, 1982.

\bibitem{Rauch1983}
S.~Rauch-Wojciechovski.
\newblock Superintegrability of the {C}alogero-{M}oser system.
\newblock {\em Phys. Lett. A}, 95:279--281, 1983.

\bibitem{RodriguezTempestaTTW}
M.~A. Rodr\'iguez and P.~Tempesta.
\newblock On the {T}remblay-{T}urbiner-{W}internitz {H}amiltonian.
\newblock In preparation.

\bibitem{Rodriguez2008}
M.~A. Rodr\'iguez, P.~Tempesta, and P.~Winternitz.
\newblock Reduction of superintegrable systems: The anisotropic harmonic
  oscillator.
\newblock {\em Phys. Rev. E}, 78:046608, 2008.

\bibitem{Shafarevich1994}
I.~R. Shafarevich.
\newblock {\em Basic Algebraic Geometry 1}, volume 213 of {\em Grundlehren der
  mathematischen Wissenschaften}.
\newblock Springer-Verlag, Berlin, Heidelberg, New York, 2 edition, 1994.

\bibitem{Tabor1989}
M.~Tabor.
\newblock {\em Chaos and {I}ntegrability in {N}onlinear {D}ynamics}.
\newblock Wiley, New York, 1989.

\bibitem{TondoTempesta2016}
G.~Tondo and P.~Tempesta.
\newblock Haantjes manifolds of classical integrable systems.
\newblock Preprint on \texttt{arXiv:1405.5118v3}.

\bibitem{TTW2009}
F.~Tremblay, A.~V. Turbiner, and P.~Winternitz.
\newblock An infinite family of solvable and integrable quantum systems on a
  plane.
\newblock {\em J. Phys. A: Math. Theor.}, 42:242001, 2009.

\bibitem{TTW2010}
F.~Tremblay, A.~V. Turbiner, and P.~Winternitz.
\newblock Periodic orbits for an infinite family of classical superintegrable
  systems.
\newblock {\em J. Phys. A: Math. Theor.}, 43:015202, 2010.

\bibitem{Tsiganov2000}
A.~V. Tsiganov.
\newblock The {D}rach superintegrable systems.
\newblock {\em J. Phys. A: Math. Theor.}, 33:7407--7422, 2000.

\bibitem{Tsiganov2008}
A.~V. Tsiganov.
\newblock Addition theorems and the {D}rach superintegrable systems.
\newblock {\em J. Phys. A: Math. Theor.}, 41:335204, 2008.

\bibitem{Whittaker}
E.~T. Whittaker.
\newblock {\em A {T}reatise on the {A}nalytical {D}ynamics of {P}articles and
  {R}igid {B}odies}.
\newblock Cambridge University Press, Cambridge, 1999.

\bibitem{WhittakerWatson1927}
E.~T. Whittaker and G.~N. Watson.
\newblock {\em A {Course} of {Modern} {Analysis}}.
\newblock Cambridge University Press, 4th edition, 1927.

\bibitem{ZakharovKuznetsov1986}
V.~E. Zakharov and E.~A. Kuznetsov.
\newblock Multi-scale expansions in the theory of systems integrable by the
  inverse scattering transform.
\newblock {\em Phys. D:}, 18:455--463, 1986.

\bibitem{ZakharovShabat1973}
V.~E. Zakharov and A.~B. Shabat.
\newblock Interaction between solitons in a stable medium.
\newblock {\em Sov. Phys. JETP}, 37:823--828, 1973.

\bibitem{Zarmi2002}
Y.~Zarmi.
\newblock The {B}ertrand theorem revised.
\newblock {\em Am. J. Phys.}, 70:446--449, 2002.

\end{thebibliography}

\end{document}